\documentclass[11pt, a4paper]{article}
\pdfoutput=1
\usepackage{authblk} 
\usepackage{amssymb}  
\usepackage{amsfonts}  
\usepackage{amsmath}  
\usepackage{amsthm}  

\usepackage{bm}  
\usepackage{bbm}  
\usepackage{romannum} 
\usepackage{comment}
\usepackage[colorlinks, allcolors=purple]{hyperref}
\usepackage{graphicx}
\usepackage{tabularx}
\usepackage{bbding}
\usepackage{enumerate}

\theoremstyle{plain}
\newtheorem{theorem}{Theorem}
\newtheorem{lemma}[theorem]{Lemma}
\newtheorem{corollary}{Corollary}[theorem]

\usepackage{tikz, datatool, ifthen, csvsimple}
\usetikzlibrary{arrows, shapes, math, decorations, decorations.markings}
\usetikzlibrary{calc}
\usepackage{pgfplots}
\pgfplotsset{compat=newest}
\pgfplotsset{
    colormap={mycm}{rgb255=(225, 225, 225) rgb255=(225, 225, 225)},
    colormap/mycm/.style={
        colormap name=mycm,
    },
}
\usepackage{pgfmath}
\input pgfmath.tex

\title{Fermi-Hubbard model on non-bipartite lattices: flux problem and emergent chirality}

\author[1]{Wayne Zheng}
\affil[1]{Department of Physics, The Ohio State University, Columbus, Ohio 43210, USA}

\date{\today}

\begin{document}
\pagenumbering{arabic}

\maketitle

\begin{abstract}
On several one dimensional (1D) and 2D non-bipartite lattices, we study both free and Hubbard interacting lattice fermions when there are some magnetic fluxes threaded or gauge fields coupled.
On one hand, we focus on finding out the optimal flux which minimizes the energy of fermions at specific fillings.
For spin-$1/2$ fermions at half-filling on a ring lattice consisting of odd-numbered sites, the optimal flux turns out to be $\pm\pi/2$.
We prove this conclusion for Hubbard interacting fermions utilizing a \emph{generalized reflection positivity} technique.
It can lead to further applications on 2D non-bipartite lattices such as triangle and Kagome.
At half-filling the optimal flux patterns on the triangular and Kagome lattice can be ascertained to be $\pm[\pi/2, \pi/2]$, $\pm[\pi/2, \pi/2, 0]$, respectively (see the meaning of these notations in the main text).
We also find that chirality emerges in these optimal flux states.
On the other hand, we verify these exact conclusions and further study some other fillings with the numerical exact diagonalization method.
It is found that, when it deviates from half-filling, Hubbard interactions can alter the optimal flux patterns on these lattices.
Moreover, both in 1D and 2D, numerically observed emergent flux singularities driven by strong Hubbard interactions in the ground states are discussed and interpreted as some kind of non-Fermi liquid features.
\end{abstract}

\tableofcontents

\section{Introduction}
\label{sec:intro}
The Fermi-Hubbard model is very famous and important~\cite{nphys2013}.
It has appealed to researchers for decades as the simplest route towards understanding strongly correlated fermionic quantum many-body systems.
It is widely believed that the Hubbard model should be closely related to the essential ingredients of Mott insulator and high-temperature superconductivity~\cite{Baskaran1987, RevModPhys.66.763, RevModPhys.78.17, anderson1997}.
Numerically it has also been studied extensively~\cite{PhysRevB.31.4403, PhysRevB.40.506, PhysRevX.5.041041}.
Recently, interest on it is stimulated again since it has been simulated by ultracold atoms in experiments~\cite{Jaksch2005, Mazurenko2017, Tarruell2018}.
Here for our theoretical interests, we would like to mention and emphasize three related aspects.

\emph{Firstly}, since the perturbation theory always cannot provide us with faithful and clear results if Hubbard interactions are sufficiently strong, rigorous theorems bring forward many great insights into the non-perturbative features in the Hubbard model~\cite{1995cond.mat.12169T}.
On a 1D bipartite lattice\footnote{A bipartite lattice $\Lambda$ is the one $\Lambda=A\cup B, A\cap B=\emptyset$ and $t_{ij}=0$ if $i, j\in A$ or $i, j\in B$~\cite{PhysRevLett.73.2158}, where $t_{ij}$ is the fermion hopping amplitude.}, it has been solved exactly and shown that there is no Mott transition~\cite{PhysRevLett.20.1445, Lieb2003}.
On 2D bipartite lattices, at half-filling E. H. Lieb settled the ground state's uniqueness and its total spin up to any finite repulsive Hubbard interactions~\cite{PhysRevLett.62.1201}.
If a hole is doped on 2D bipartite lattices, strong Hubbard interactions can induce an emergent Nagaoka ferromagnetism~\cite{PhysRev.147.392}.
We notice that, both in 1D and 2D, the \emph{bipartiteness} plays an important role in many of these significant theorems.
It leads to a special kind of particle-hole transformation, where a minus sign only follows on one of the two bipartite subsets of lattice sites.
Moreover, quantum Monte Carlo also can avoid the severe sign problem due to the bipartiteness~\cite{Meng2010, Sorella2012}.
Thus a natural question can be asked: \emph{Why does the bipartiteness seem to be so essential here? What will happen if we lose it?}

\emph{Secondly}, without any doubt fermions and gauge fields indeed can emerge from very different strongly correlated bosonic quantum many-body systems~\cite{Anderson393, PhysRevB.37.580, PhysRevB.46.5621, PhysRevB.67.245316, rsta.2015.0248}.
Exactly solvable Kitaev's honeycomb spin model is supposed to be the most convincing example, which equivalently turns out to be emergent free Majorana fermions coupled to a $\mathbbm{Z}_{2}$ lattice gauge field~\cite{Kitaev2006}.
These emergent gauge fields living on the lattice links can form magnetic fluxes, of which the corresponding effective magnetic field can be so strong that no experiments can realize it on the earth.
The low energy gauge fluctuations above the mean-field state turn out to be crucial and even the topology of the gauge fields plays a significant role~\cite{PhysRevB.37.3774, wen2004quantum, rsta.2015.0248, Sachdev2018}.
In these kinds of fermion-gauge field coupled systems, finding out the optimal flux pattern to minimize the ground state energy at zero-temperature or statistical free energy at any finite temperature is called the \emph{flux problem}.
For Hubbard interacting fermions at half-filling, E. H. Lieb solved the flux problem on generic 2D bipartite lattices with the help of an elegant technique called \emph{reflection positivity}~\cite{PhysRevLett.73.2158, Macris1996} (RP) which was first introduced in the quantum field theory~\cite{Osterwalder1978}.
Lieb's result directly leads to the solution of Kitaev's honeycomb model.
The optimal $\pi$-flux Dirac state on a square lattice has been observed numerically in a fermion-gauge fields coupled system~\cite{Gazit2017}.
It is used to serve as a good starting point to construct quantum spin liquids (QSLs) in the language of fermionic partons~\cite{PhysRevB.65.165113, PhysRevX.8.011012}.
Note that days earlier, high-$T_{c}$ superconductivity is also found to be closely related to the flux issue~\cite{PhysRevB.37.3664, PhysRevB.37.3774}.

\emph{Thirdly}, it is well known that the spin chiral operator $\chi\equiv\mathbf{\sigma}_{1}\cdot(\mathbf{\sigma}_{2}\times\mathbf{\sigma}_{3})$ can be expressed by the flux Berry phase $\phi$ acquired by fermions hopping along a closed plaquette~\cite{PhysRevB.39.11413}.
To be specific, $\langle\chi\rangle\propto\sin\phi$, where $\phi$ is the flux threaded throughout the plaquette.
For 2D bipartite lattices, the typical $0$ or $\pi$-flux optimal states are non-chiral, where chirality $\langle\chi\rangle$ vanishes.
Therefore, there does not exist persistent spin current around the plaquettes induced only by non-zero chirality.

In this sense, in this paper we would like to investigate Hubbard interacting fermions without bipartiteness any longer, to see its interplay with gauge fields and chirality.
We obtained several new results analytically as well as numerically.
The rest of this paper is organized as follows.
In Sec.~\ref{sec:1d}, from non-interacting to interacting cases, 1D lattice fermions are investigated.
The optimal flux for spin-$1/2$ fermions at half-filling on a non-bipartite odd-numbered ring is proved and verified numerically no matter the Hubbard interactions are present or not.
In Sec.~\ref{sec:2d}, we generalize our technique to 2D and study the flux problem for the Hubbard model on 2D non-bipartite lattices.
At half-filling, the optimal flux patterns for the triangular and Kagome lattices can be nailed down.
However, when it deviates from half-filling, there is no rigorous analytical results anymore.
Some numerical results are provided and discussed in both 1D and 2D.
In particular, emergent flux singularities driven by strong Hubbard interactions are addressed and identified as some non-Ferm liquid (NFL) features.
In Sec.~\ref{sec:summary}, we end up with a brief summary and discussion.

\section{1D lattice}
\label{sec:1d}
\subsection{Non-interacting spin-$1/2$ fermions}
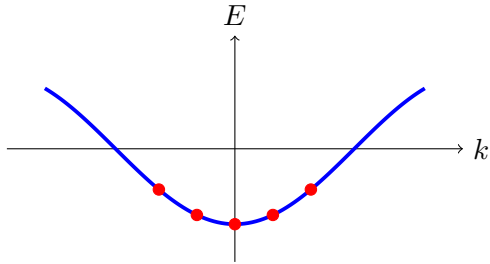
\begin{figure}[!ht]
    \centering
    \begin{tikzpicture}[]
        \def\r{0.075cm};
        \pgfmathsetmacro\de{0.0};
        \pgfmathsetmacro\xshift{1.5*sqrt(3)};
        \draw[->] (-3.0, 0) -- (3.0, 0) node[right] {$k$};
        \draw[->] (0, -1.5) -- (0, 1.5) node[above] {$E$};
        \draw[samples=100, scale=1.0, domain=-2.5:2.5, variable=\k, color=blue, line width=0.05cm] plot ({\k}, {-cos(\k/pi*180)});
        \draw[fill, color=red] ({0}, {-cos(0)}) circle (\r);
        \draw[fill, color=red] ({0.5}, {-cos(0.5/pi*180)}) circle (\r);
        \draw[fill, color=red] ({-0.5}, {-cos(-0.5/pi*180)}) circle (\r);
        \draw[fill, color=red] ({1.0}, {-cos(1.0/pi*180)}) circle (\r);
        \draw[fill, color=red] ({-1.0}, {-cos(-1.0/pi*180)}) circle (\r);
    \end{tikzpicture}
    \caption{Band of free fermions in 1D.}
   \label{fig:1d_band}
\end{figure}

In the first place, let us warm up by taking a look at the simplest case, namely fermions living on a 1D lattice, which is always assumed to form a ring.
For two branches of non-relativistic spin-$1/2$ free non-interacting fermions $\sigma=\uparrow, \downarrow$, they can be treated separately as
\begin{equation}
    H_{t}=-\sum_{\sigma}\sum_{j=0}^{L-1}\left(t_{j, j+1} c_{j\sigma}^{\dagger}c_{j+1, \sigma}+h.c. \right).
    \label{eq:ham_free_spinful}
\end{equation}
As usual, $j+1\equiv (j+1)\mod L$.
$\{c_{i, \sigma}^{\dagger}, c_{j, \sigma^{\prime}}\}=\delta_{i,j}\delta_{\sigma, \sigma^{\prime}}$ and $\{c_{i, \sigma}, c_{j, \sigma^{\prime}}\}=\{c_{i, \sigma}^{\dagger}, c_{j, \sigma^{\prime}}^{\dagger}\}=0$ define the complex fermionic operators.
$t_{j, j+1}$ is the Wannier hopping amplitude and $|t_{j, j+1}|=t=1.0$ is set to be energy unit throughout this paper.
A magnetic flux $\phi$ can be added through appropriate boundary conditions such as $t_{N-1, 0}=e^{-\text{i}\phi}$ while the others $t_{j,j+1}=1.0, j\neq L-1$.
By a discrete Fourier transformation $c_{j}=\frac{1}{\sqrt{L}}\sum_{k}e^{\text{i}kj}c_{k}$, $H_{t}=-\sum_{\sigma}\sum_{k}(2\cos{k})c_{k, \sigma}^{\dagger}c_{k, \sigma}$, where $k=(2\pi l+\phi)/L$ is constrained by the boundary condition.
$l\in\mathbbm{Z}$.
The ground state energy is determined by the band fillings given $\text{U}(1)$ conserved particle numbers $N_{\uparrow, \downarrow}$, which is illustrated in Figure~\ref{fig:1d_band}.
It is easy to check that, when $N_{\uparrow}=N_{\downarrow}=2n, n\in\mathbbm{Z}_{+}$, the optimal flux takes $\phi_{\text{opt}}=\pi$.
When $N_{\uparrow}=N_{\downarrow}=2n+1, n\in\mathbbm{Z}_{+}$, the optimal flux is $\phi_{\text{opt}}=0$.
When $N_{\uparrow}=2n, N_{\downarrow}=2m-1, n, m\in\mathbbm{Z}_{+}$, the optimal flux should lie at some value between $0$ and $\pi$ to minimize the filling energy of these two branches of fermions.
According to the discussion in the Appendix~\ref{sec:1d_spinless} if we implement the Jordan-Wigner transformation on these two branches of fermions separately, different parities would lead to a competition when it comes to minimizing the ground state energies by the natural inequality~\cite{PhysRevLett.111.100402, PhysRevB.97.125153}.
Therefore, a chiral optimal flux indeed can emerge in such a scenario.
Generally, it depends on $N_{\uparrow, \downarrow}$ and $L$ as $\phi_{\text{opt}}=\phi_{\text{opt}}(N_{\uparrow}, N_{\downarrow}, L)$.
The optimal flux in this scenario can be determined by certain transcendental triangular equation, which can be solved numerically.
For example, if $N_{\uparrow}=2, N_{\downarrow}=1$, at a local minimum we have
\begin{equation}
    2\sin\left(\frac{\phi}{L}\right)=\sin\left(\frac{2\pi-\phi}{L}\right).
    \label{eq:optimal_flux_filling21}
\end{equation}
Say, $L=5$, here the optimal flux $\phi_{\text{opt}}\simeq 1.9536$.
However, we found that $\phi_{\text{opt}}$ is so much special when it comes to the half-filled case, which is \emph{independent} of the odd $L$.
For simplicity, here we only focus on the cases with a minimal $|S_{\text{tot}}^{z}|=\frac{1}{2}|N_{\uparrow}-N_{\downarrow}|$ on a non-bipartite odd-numbered ring.
Then we have the following lemma:
\begin{lemma}
    \label{lem:free_spinful}
    For spin-$1/2$ free non-interacting fermions with a minimal $|S_{\text{tot}}^{z}|$ on a non-bipartite odd-numbered ring at half-filling, the optimal fluxes for the ground states are $\pm\pi/2$, which are independent of the lattice size $L$.
\end{lemma}
\begin{proof}
    See Appendix~\ref{app:half_filled_odd_ring}. 
\end{proof}
For finite temperature, we can prove that
\begin{lemma}
    For spin-$1/2$ free fermions on a ring lattice defined by Eq.~(\ref{eq:ham_free_spinful}), if the parities of particle numbers $N_{\uparrow, \downarrow}$ are identical, at finite temperature the optimal flux for the statistical free energy $F$ is $0$ or $\pi$ depending on the parity is odd or even, respectively.
\end{lemma}
\begin{proof}
The basis for spin-$1/2$ fermions spanning the Hilbert space can be written in a specific representation~\cite{PhysRevLett.62.1201} $|\alpha\rangle_{\uparrow}\otimes|\gamma\rangle_{\downarrow}$.
Expanding the canonical partition function like
\begin{equation}
    Z=\text{tr}\left( e^{-\beta H_{t}} \right)=\lim_{M\rightarrow\infty}\text{tr}\left[ V^{M}(\phi) \right],
    \label{eq:}
\end{equation}
and
\begin{equation}
    V(\phi)=1+\delta\sum_{\sigma}\left( \sum_{j=0}^{L-2}c_{j\sigma}^{\dagger}c_{j+1, \sigma}+e^{\text{i}\phi}c_{L-1, \sigma}^{\dagger}c_{0, \sigma}+h.c. \right),
    \label{eq:}
\end{equation}
where $\delta=\beta/M$.
We rewrite $V^{M}(\phi)=\sum_{\alpha}X^{\alpha}=\sum_{\alpha}\prod_{\sigma}X_{\sigma}^{\alpha}$ when we rearrange the operator string by their spin indices.
Then in this representation we have $\text{tr}\left( \prod_{\sigma}X_{\sigma} \right)=\text{tr}\left(X_{\uparrow}\right)\cdot\text{tr}(X_{\downarrow})$.
The lowest order nontrivial operator strings take the form as $\text{tr}\left( X_{\uparrow} \right)\cdot\text{tr}\left( \mathbbm{1}_{\downarrow} \right)+\text{tr}\left( \mathbbm{1}_{\uparrow} \right)\cdot\text{tr}\left( X_{\downarrow} \right)=(-)^{N_{\uparrow}-1}\delta^{L}e^{\text{i}\phi}D_{\downarrow}+(-)^{N_{\downarrow}-1}\delta^{L}e^{\text{i}\phi}D_{\uparrow}+h.c.=2(-)^{N_{\uparrow}-1}D_{\downarrow}\delta^{L}\cos\phi+2(-)^{N_{\downarrow}-1}D_{\uparrow}\delta^{L}\cos\phi$, where $D_{\uparrow, \downarrow}=C_{L}^{N_{\uparrow, \downarrow}}$ is the dimension of sub-Hilbert space corresponding to spin-$\uparrow$ and -$\downarrow$ fermions, respectively.
If $N_{\uparrow}$ and $N_{\downarrow}$ share the same parity, this very term maximizes as same as the free spinless fermions discussed in Appendix~\ref{sec:1d_spinless}.
The higher-order crossed term such as $+2\delta^{2L}\cos(2\phi)$ maximizes at the same time.
Once the partition function is maximized, free energy $F=-\frac{1}{\beta}\ln Z$ is minimized.
\end{proof}

If the parities of $N_{\uparrow}$ and $N_{\downarrow}$ are different, at finite temperature there will be some competition in nontrivial terms such as $\pm2(D_{\downarrow}-D_{\uparrow})\cos\phi$, which maximizes at $0$ or $\pi$ while the crossed term $-2\cos(2\phi)$ maximizes at $\pi/2$.
Therefore, to determine the optimal flux for the finite temperature free energy $F$ is hard.
It might differ from the ground state.

\subsection{Turn on Hubbard interactions}
\begin{figure}[!ht]
    \centering
    \includegraphics[width=0.85\textwidth]{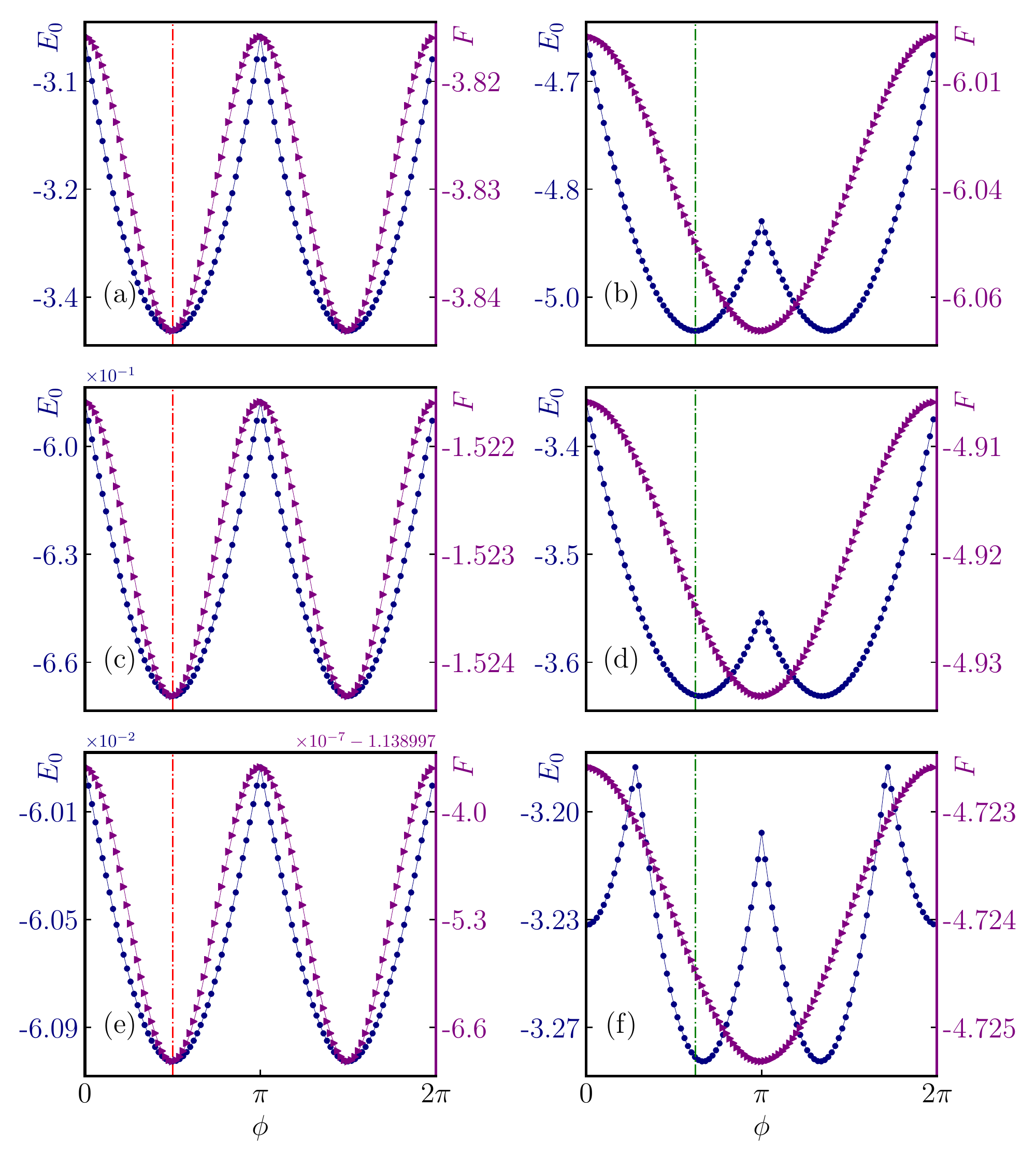}
\caption{Ground state energy $E_{0}$ and the finite temperature free energy $F$ of the 1D Hubbard model on a ring with fixed fillings $N_{\uparrow}=2, N_{\downarrow}=1$.
Vertically from top to bottom, (a, b), (c, d) and (e, f) denote $U/t=0.0, 10.0, 100.0$, respectively.
Horizontally from left to right, (a, c, e) and (b, d, f) denote the lattice size $L=3, 5$, respectively.
Free energy is computed at $\beta=1.0$.
The red dashed line marks the optimal flux $\phi_{\text{opt}}=\pi/2$ for the ground state energy at half-filling.
Green dashed line marks the optimal flux $\phi_{\text{opt}}\simeq 1.9536$ for the ground state energy away from half-filling.}
    \label{fig:flux_ring}
\end{figure}

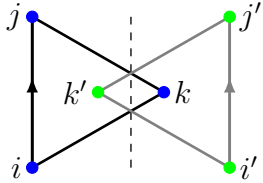
\begin{figure}[!ht]
    \centering
    \begin{tikzpicture}[x=1cm, y=1cm]
        \tikzset{
            arrow/.style = {
                decoration={
                    markings, 
                    mark=at position 0.6 with{\arrow{latex}}
                },
                postaction={decorate}
            }
        }
        \def\r{0.075cm};
        \def\l{2.0cm};
        \pgfmathsetmacro\de{0.0};
        \pgfmathsetmacro\xshift{1.5*sqrt(3)};
        \draw[dashed, line width=0.5pt] ({0.5*\xshift}, {0.0-\de}) -- ({0.5*\xshift}, {2.0+\de});
        \draw[line width=1pt] (0.0, 0.0) -- (0.0, \l) -- ({\l*cos(30)}, {\l*sin(30)}) -- (0.0, 0.0);
        \draw[arrow, line width=1pt] (0.0, 0.0)--(0.0, \l); 
        \draw[fill, color=blue] (0.0, 0.0) circle (\r);
        \node[anchor=east] at  (0.0, 0.0) {\large $i$};
        \draw[fill, color=blue] (0.0, \l) circle (\r);
        \node[anchor=east] at (0.0, \l) {\large $j$};
        \draw[fill, color=blue] ({\l*cos(30)}, {\l*sin(30)}) circle (\r);
        \node[anchor=west] at ({\l*cos(30)}, {\l*sin(30)}) {\large $k$};
        \draw[color=gray, line width=1pt] (\xshift, 0.0) -- (\xshift, \l) -- ({\xshift-sqrt(3)}, {\l*sin(30)}) -- (\xshift, 0.0);
        \draw[color=gray, arrow, line width=1pt] (\xshift, 0.0)--(\xshift, \l); 
        \draw[fill, color=green] (\xshift, 0.0) circle (\r);
        \node[anchor=west] at  (\xshift, 0.0) {\large $i^{\prime}$};
        \draw[fill, color=green] (\xshift, \l) circle (\r);
        \node[anchor=west] at (\xshift, \l) {\large $j^{\prime}$};
        \draw[fill, color=green] ({\xshift-sqrt(3)}, {\l*sin(30)}) circle (\r);
        \node[anchor=east] at ({\xshift-sqrt(3)}, {\l*sin(30)}) {\large $k^{\prime}$};
   \end{tikzpicture}
   \caption{
     $\{i, j, k\}$ represents a half-filled Hubbard model living on a $L=3$ ring lattice.
   $\{i^{\prime}, j^{\prime}, k^{\prime}\}$ is merely a fictitious reflection of $\{i, j, k\}$ along the dashed line.
   Arrows mark the flux accumulating directions.}
   \label{fig:double_triangle}
\end{figure}

When the simplest kind of on-site interaction 
\begin{equation}
    H_{U}=U\sum_{j}n_{j\uparrow}n_{j\downarrow},
    \label{eq:ham_hubbard}
\end{equation}
is turned on, we have the so-called Hubbard model~\cite{PhysRevLett.10.159, Hubbard1963}.
Its lattice Hamiltonian is given by $H=H_{t}+H_{U}$, where repulsive $U>0$ is always assumed throughout this paper.
The two branches of fermions gradually begin to interact and get entangled with each other as $U$ increases from zero.
On a bipartite ring lattice at half-filling, we can expect a four-fold degeneracy at most in terms of free spin-$1/2$ fermions, since every branch can contribute a two-fold degeneracy.
Recall that E. H. Lieb told us that any finite Hubbard $U$ can split this degeneracy thereby leave over a unique ground state~\cite{PhysRevLett.62.1201}.
For $N=N_{\uparrow}+N_{\downarrow}=2n, n\in\mathbbm{Z}_{+}$, the optimal flux for the ground state of the Hubbard model has been proved~\cite{nakano2000} to be $0$ or $\pi$ depending on the parity of $N/2$.
Thus, it is quite meaningful to ask, on non-bipartite lattices, how do the Hubbard interactions impact on comprehensive features of free fermions, including the optimal flux problem.

\subsubsection{A numerical example}
Above all, let us carry out some numerical experiments by exact diagonalization (ED) technique utilizing the \texttt{ARPACKPP} package~\cite{arpackpp}.
They are quite simple but very helpful to obtaining some basic intuitions.
Let us only consider three fermions $N_{\uparrow}=2, N_{\downarrow}=1$.
On one hand, $L=3$ means half-filling, as illustrated in Figure~\ref{fig:flux_ring}(a, c, e) with increasing $U/t=0.0, 10.0, 100.0$, both the ground state energy $E_{0}$ and finite temperature free energy $F$ always minimize at $\phi_{\text{opt}}=\pm\pi/2$.
Even very strong Hubbard interaction still does not affect the optimal flux value.

On the other hand, as we can see in Figure~\ref{fig:flux_ring}~(b, d, f), $L=5$ means it is not half-filled.
Firstly, the optimal flux for the ground state can be altered by the Hubbard interactions.
There does not exist a universal optimal flux for the ground state of the Hubbard model when it is not half-filled.
As $U/t$ increases, the optimal flux for the ground state gradually shifts from the free fermions' $\phi_{\text{opt}}\simeq 1.9536$, which is approximately given by Eq.~(\ref{eq:optimal_flux_filling21}), to $2\pi/3$.
In view of Eq.~(\ref{eq:optimal_flux_filling21}), it is interesting to realize that $2\pi/3$ is nothing but the optimal flux solution for the ground state of these free fermions when $L\rightarrow\infty$.
These two limits $U\rightarrow\infty$ (fixed finite $L$) and $L\rightarrow\infty$ ($U=0.0$) somehow arrive at the same optimal flux.
This implies, in a quantum many-body system strong interactions indeed can drive some emergent non-pertubative features which can not be understood by free or weakly interacting pictures.
Secondly,
the finite temperature free energy $F$ does not share the identical optimal flux with the ground state anymore.
Its optimal flux locates at $\phi=\pi$.

This numerical test, as well as the inspiration of Lemma~\ref{lem:free_spinful}, give us faith that the optimal fluxes $\phi_{\text{opt}}=\pm\pi/2$ always hold for the half-filled Hubbard model sitting on a non-bipartite odd-numbered ring.
Half-filling seems to be a special fixed point.
However, when it deviates from half-filling, it is much more complicated and there seemingly does not exist a universal conclusion.

\subsubsection{Generalized reflection positivity}
For the Hubbard interacting fermions, note that the RP technique only can be applied to a ring comprised of a even number of sites, hence resulting in the optimal flux of $0$ or $\pi$.
We succeeded in proving the following theorem with the aid of a \emph{generalized reflection positivity} (GRP) technique.
\begin{theorem}
    For a half-filled repulsive Hubbard model with a minimal $|S_{\text{tot}}^{z}|$ on a non-bipartite odd-numbered ring, at any finite temperature the optimal fluxes for its free energy $F$ are $\pm\pi/2$.
\end{theorem}
\begin{proof}
Here we take the simplest case to explain the GRP, as denoted by Figure~\ref{fig:double_triangle}, where a half-filled Hubbard model lives on a triangle $\{i, j, k\}$.
We would like to make a fictitious symmetric reflective copy of the system $\{i, j, k\}$ to $\{i^{\prime}, j^{\prime}, k^{\prime}\}$.
As same as discussed in Ref.~\cite{PhysRevLett.73.2158}, we are at liberty to choose the gauge as the flux is only added on the non-intersected links $(i, j)$ and $(i^{\prime}, j^{\prime})$.
The other intersected links can be set to be $t_{ik}=t_{jk}=t$ although generally they are not $\Theta$ invariant, in which $\Theta$ is comprised of three steps: geometric reflection $\mathcal{R}$ followed by a particle-hole transformation and a complex conjugation $\mathcal{C}$.
If we regard these six sites as a whole system, according to E. H. Lieb's theorem~\cite{PhysRevLett.73.2158} which is also reviewed in Appendix~\ref{app:reveew_RP}, the fulfillment of $\Theta(t_{ij}c_{i\sigma}^{\dagger}c_{j\sigma}+h.c.)=-t_{j^{\prime}i^{\prime}}c_{j^{\prime}\sigma}^{\dagger}c_{i^{\prime}\sigma}+h.c.=-\mathcal{R}(t_{ij}c_{i\sigma}^{\dagger}c_{j\sigma}+h.c.)$ leads to the maximum of the partition function of the Hubbard model.
Note that $\{i^{\prime}, j^{\prime}, k^{\prime}\}$ is merely a fictitious reflective image of the original system.
If we would like to separate the whole system and view them as two equivalent ones, we shall make another followed complex conjugation $t_{i^{\prime}j^{\prime}}=-t_{ij}\xrightarrow[]{\mathcal{C}}-t_{ij}^{*}\equiv t_{ij}$, which means $t_{ij}$ is pure imaginary thereupon a $\pi/2$ flux is threaded through the triangle.
Note that the second complex conjugation carried out here is to flip the flux direction of the reflective mirror system $\{i^{\prime}, j^{\prime}, k^{\prime}\}$ back as to match the original system $\{i, j, k\}$.
It is easy to utilize our GRP to other ring lattices with odd-numbered sites $L>3$ as every non-intersected link contributes a $\pm\pi/2$ gauge flux.
When the partition function is maximized, the free energy $F$ is minimized.
\end{proof}
In addition, we have some remarks.
\begin{enumerate}[i]
\item
We think that it is dangerous to rashly deduce the ground state properties as to let $\beta\rightarrow\infty$.
A possible example is illustrated in Figure~\ref{fig:flux_ring}(b, d, f) that the thermal free energy does not share the same optimal flux with the ground state when it deviates from half-filling.
There possibly exists a phase transition when the temperature jumps from any finite value to zero.
While for the half-filling, this somehow turns out to be safe such as the half-filled examples numerically shown in Figures~\ref{fig:flux_ring},~\ref{fig:flux_bowtie} and~\ref{fig:flux_kagome}.
There the ground state energy always shares the identical optimal flux with the corresponding finite temperature free energy, at least speaking in terms of these numerical samples.
We think it should have something to do with the coincidence implied by Lemma~\ref{lem:free_spinful}, which suggests that a half-filled non-bipartite odd-numbered ring happens to be an unrenormalized fixed point, where the optimal flux cannot be altered and renormalized by the Hubbard interactions any longer.
Let us take a closer look at the finite temperature partition function of 1D Hubbard model away from half-filling.
Still suppose the basis spanning the Hilbert space is written in the representation $|\alpha\rangle_{\uparrow}\otimes|\gamma\rangle_{\downarrow}$, similarly,
\begin{equation}
    Z=\text{tr}\left( e^{-\beta H} \right)=\lim_{M\rightarrow\infty}\text{tr}\left[ V^{M}(\phi) \right],
    \label{eq:}
\end{equation}
and
\begin{equation}
    \begin{aligned}
        V(\phi)=1+\delta\sum_{\sigma}&\left( \sum_{j=0}^{L-2}c_{j\sigma}^{\dagger}c_{j+1, \sigma}+e^{\text{i}\phi}c_{L-1, \sigma}^{\dagger}c_{0\sigma}+h.c.-\frac{U}{t}\sum_{j=0}^{L-1}n_{j\uparrow}n_{j\downarrow} \right).
    \end{aligned}
    \label{eq:}
\end{equation}
We write $V^{M}(\phi)=\sum_{\alpha}X^{\alpha}=\sum_{\alpha}\prod_{\sigma}X_{\sigma}^{\alpha}$ as we rearrange the operator string by their spin indices.
Then in Lieb's representation, we have $\text{tr}\left(X\right)=\text{tr}\left( \prod_{\sigma}X_{\sigma} \right)=\text{tr}\left(X_{\uparrow}\right)\cdot\text{tr}(X_{\downarrow})$.
The lowest order nontrivial term with interactions is $(-U/t)\delta\cdot\delta^{L}\cos\phi$ together with the purely kinetic contribution term $\delta^{L}\cos\phi$ we have $[1+(-U/t)\delta]\delta^{L}\cos\phi$, which maximizes as the same way as free fermions at any finite temperature.
That is, at finite temperature, the optimal flux for the free energy of the 1D Hubbard model should not be affected by any finite Hubbard interaction.
Numerical examples also confirm this as shown in Figure~\ref{fig:flux_ring}(b, d, f).
However, the optimal flux for the zero-temperature ground state energy is altered by increasing the Hubbard interactions when it is not half-filled.
Note that if $U\rightarrow\infty$, the above statement is not valid anymore since the sign of $[1+(-U/t)\delta]$ is not well-defined any longer.
One counter-example we have already know is the Nagaoka state~\cite{PhysRev.147.392} saying that Hubbard model with one fermion away from half-filling will fully polarize in the limit $U\rightarrow\infty$.
Long-range hoppings or 2D Hubbard model can induce the Nagaoka polarization with large but finite $U/t$~\cite{Farka2014}.

\item
Because of the GRP, the original system and the fictitious reflective system share the same flux to reach the maximum of the partition function at the same time.
Thus the flux period reduces from $2\pi$ to $\pi$.
Numerically we can also see this in Figure~\ref{fig:flux_ring}(a, c, e).
$\phi=\pm\pi/2$ both are optimal fluxes but have opposite chiralities.
The spin chiral order operator $\chi=\mathbf{\sigma}_{1}\cdot(\mathbf{\sigma}_{2}\times\mathbf{\sigma}_{3})$, which depends on the imaginary part of the gauge invariant Berry phases, will be not only nonzero but also maximized with respect to the optimal $\pm\pi/2$ fluxes.
At the same time, nonvanishing charge currents $J_{ij}=\text{i}\sum_{\sigma}(c_{i\sigma}^{\dagger}c_{j\sigma}-h.c.)$~\cite{PhysRevB.44.6909, PhysRevB.98.165102} should be observed around the loop lattice.
Both time-reversal symmetry $\mathcal{T}$ and parity symmetry $\mathcal{P}$ break spontaneously and chiral order emerges.
However, the combined $\mathcal{PT}$-symmetry is not broken, which illustrates a kind of nonrelativistic PT theorem if $\text{U}(1)$ symmetry is preserved~\cite{PhysRevB.93.094437}.

\item
This result reminds us the Haldane's honeycomb model~\cite{PhysRevLett.61.2015} where there is a chiral $\pm\pi/2$ flux threaded through each second-neighbor triangle.
It is energetically favored and time-reversal symmetry $\mathcal{T}$ and parity symmetry $\mathcal{P}$ are spontaneously broken.
\end{enumerate}

\subsection{Flux singularity as a Luttinger liquid signature}
Another important issue to be addressed here is that as shown in Figure~\ref{fig:flux_ring}, the ground state energies always exhibit a non-analytical singularity at $\phi=\pi$.
This can be understood easily at least in a free particle picture.
As shown in Figure~\ref{fig:1d_band}, in the momentum space since we have a discrete momentum step $\Delta k=2\pi/L$, when the threaded flux $\phi$ is approaching $\pi$, the momentum shift behaves like $\delta k=\frac{\phi}{L}\rightarrow \frac{\pi}{L}=\frac{1}{2}\Delta k$.
Therefore, there will be a sudden rearrangement of the particle fillings in terms of the two branches of fermions when $\phi$ steps across $\pi$.
The finite-size energy gap closes at the same time.
Note that this particular phenomenon stems from nothing but the very special dimensionality of 1D.
In 1D Fermi surface shrinks to two disconnected Fermi points, which essentially make 1D fermions a NFL.
As is well known, the low energy physics of 1D fermions is described by the Luttinger liquid (LL) theory~\cite{tomonaga1950, luttinger1963, haldane1981}, which crucially differs from higher dimensions.
There exist forbidden scattering regions and unrenormalizable gapless Fermi momenta~\cite{haldane1981, PhysRevLett.79.1110, Caux2017}.
In this sense, when we vary the threaded flux $\phi$, these singularities can be regarded as a kind of Luttinger-like NFL smoking-gun evidence~\cite{10.1146, RevModPhys.79.1015}.
As expected in LL, this flux singularity locating at $\pi$ is unchanged even involving Hubbard interactions.
When it is away from half-filling as shown in Figure~\ref{fig:flux_ring}(b, d, f), the flux singularity at $\phi=\pi$ always locates there.
Only its relative energy is changed by increasing Hubbard interactions.
Moreover, we find the more interesting thing is that, as shown in Figure~\ref{fig:flux_ring}(f), with very large $U/t=100.0$, sufficiently strong Hubbard interactions can lead to some emergent flux singularities.
We attribute them to the emergent Luttinger-like flux singularities driven by strong Hubbard interactions, which are absent in free or weakly interacting fermionic states.
These states also cannot be adiabatically connected to the non-interacting fermions.
As we can see, from the perspective of flux issue, doping a half-filled Mott insulator with sufficiently large Hubbard interactions turns out making the system dramatically different from the original parent~\cite{RevModPhys.78.17}.
These emergent flux singularities can only appear in doped cases away from half-filling.

\section{2D lattices}
\label{sec:2d}
Mostly our physical interests lie in 2D.
High-$T_{c}$ superconductivity generally is regarded as a 2D physical problem, where electrons are restricted within each single Cu-O plane~\cite{RevModPhys.66.763, RevModPhys.78.17}.
Free or weakly interacting fermions in 2D are described by Fermi-liquid theory, which is dramatically different from LL theory.
There is a continuous Fermi surface and no singular scattering occurs.
As a matter of fact, decades ago it has already been proposed by P. W. Anderson that strong interactions can lead to some Luttinger-like features in 2D Hubbard model stemming from the unrenormalizable quantum phase shift and singular scatterings~\cite{PhysRevLett.64.1839, PhysRevLett.65.2306}.
When interactions are increased to sufficiently strong, Fermi-liquid theory will break down and singular scatterings within the Brillouin zone may occur.

\subsection{A trial on a bowtie lattice}
\begin{figure}[!ht]
    \centering
        \tikzset{
            hexa/.style={shape=regular polygon, regular polygon sides=6, minimum size=2, line width=0.5pt, draw, anchor=south, shape border rotate=90}
        }
        \begin{tikzpicture}[x=0.75cm, y=0.75cm]
        \tikzset{
            arrow/.style = {
                decoration={
                    markings, 
                    mark=at position 0.6 with{\arrow{latex}}
                },
                postaction={decorate}
            }
        }
        \begin{scope}[shift={(-2.0, 0.0)}]
        \def\r{0.1};
        \def\l{1.5};
        \pgfmathsetmacro\de{0.0};
        \draw[dashed, line width=0.5pt] ({0.75*\l}, {-0.5*\l}) -- ({0.75*\l}, {sqrt(3)*\l+0.5*\l});
        \draw[line width=1pt] (0.0, 0.0) -- (\l, 0.0) -- ({\l*sin(30)}, {\l*cos(30)}) -- (\l, {2*\l*cos(30)}) -- (0.0, {2*\l*cos(30)}) -- ({\l*sin(30)}, {\l*cos(30)}) -- (0.0, 0.0);
        \draw[arrow, line width=0.5pt] (0.0, 0.0) -- ({\l*sin(30)}, {\l*cos(30)});
        \draw[arrow, line width=0.5pt] ({\l*sin(30)}, {\l*cos(30)}) -- (0.0, {2*\l*cos(30)});
        \draw[fill, color=blue] (0.0, 0.0) circle (\r);
        \draw[fill, color=blue] (\l, 0.0) circle (\r);
        \draw[fill, color=blue] ({\l*sin(30)}, {\l*cos(30)}) circle (\r);
        \draw[fill, color=blue] (\l, {2*\l*cos(30)}) circle (\r);
        \draw[fill, color=blue] (0.0, {2*\l*cos(30)}) circle (\r);
        \end{scope}
        \begin{scope}[shift={(-1.25, 0.0)}]
        \def\r{0.1};
        \def\l{1.5};
        \pgfmathsetmacro\de{0.0};
        \draw[color=gray, line width=1pt] (0.0, 0.0) -- (\l, 0.0) -- ({\l*sin(30)}, {\l*cos(30)}) -- (\l, {2*\l*cos(30)}) -- (0.0, {2*\l*cos(30)}) -- ({\l*sin(30)}, {\l*cos(30)}) -- (0.0, 0.0);
        \draw[color=gray, arrow, line width=0.5pt] (\l, 0.0) -- ({\l*sin(30)}, {\l*cos(30)});
        \draw[color=gray, arrow, line width=0.5pt] ({\l*sin(30)}, {\l*cos(30)}) -- (\l, {2*\l*cos(30)});
        \draw[fill, color=green] (0.0, 0.0) circle (\r);
        \draw[fill, color=green] (\l, 0.0) circle (\r);
        \draw[fill, color=green] ({\l*sin(30)}, {\l*cos(30)}) circle (\r);
        \draw[fill, color=green] (\l, {2*\l*cos(30)}) circle (\r);
        \draw[fill, color=green] (0.0, {2*\l*cos(30)}) circle (\r);
        \end{scope}
\end{tikzpicture}
\caption{Carrying out the GRP along the dashed line on a bowtie lattice.
Its fictitious reflection copy is illustrated by gray links.}
   \label{fig:bowtie}
\end{figure}
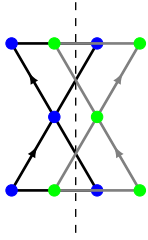

\begin{figure}[!ht]
    \centering
    \includegraphics[width=0.85\textwidth]{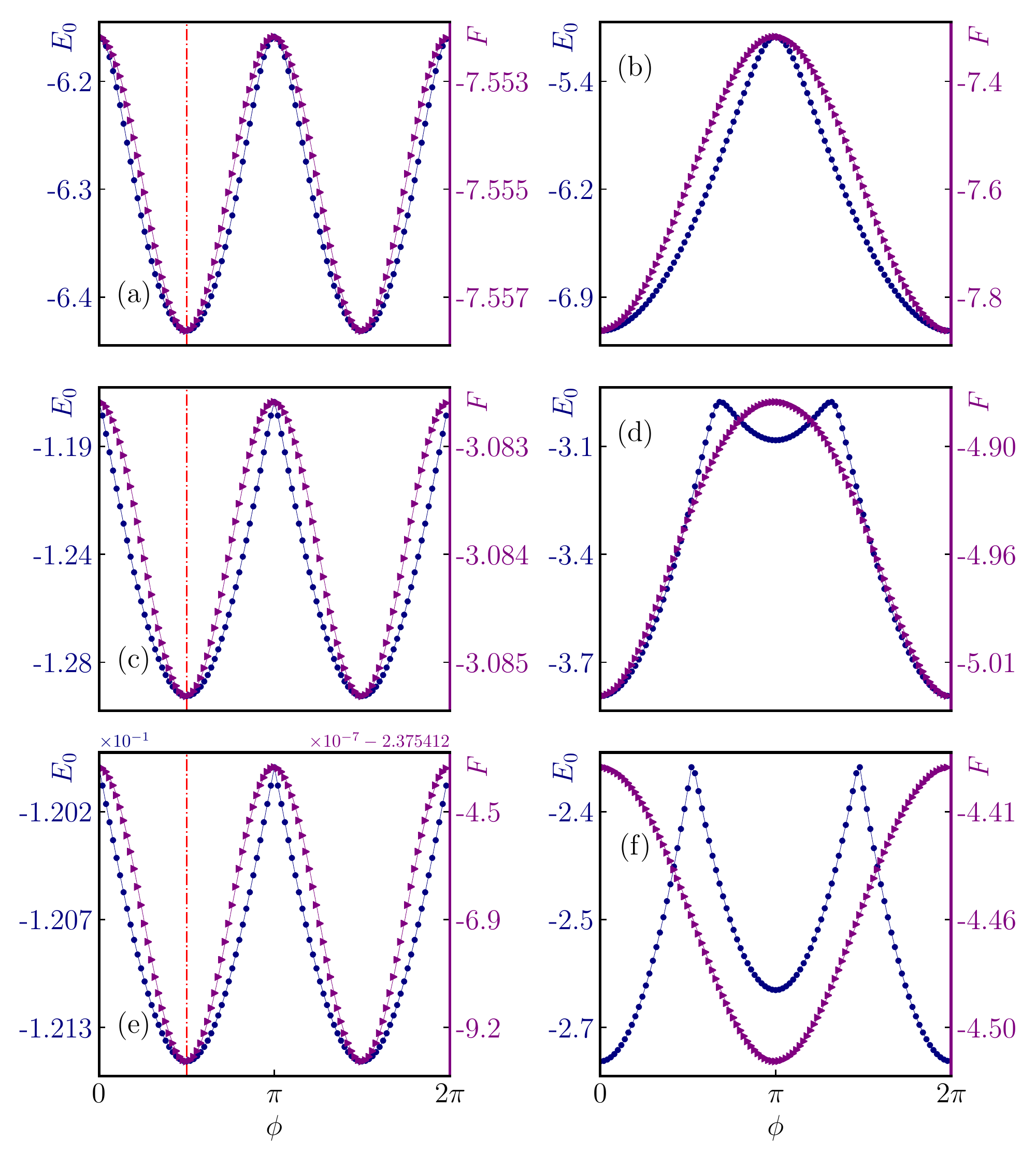}
    \caption{Ground state energy $E_{0}$ and the finite temperature free energy $F$ of the Hubbard model on a bowtie lattice. Vertically, (a, b), (c, d) and (e, f) denote $U/t=0.0, 10.0, 100.0$, respectively. Horizontally, (a, c, e) denote half-filling $N_{\uparrow}=3, N_{\downarrow}=2$. (b, d, f) denote filling $N_{\uparrow}=N_{\downarrow}=2$. Free energy is computed at $\beta=1.0$. Red dashed line marks the optimal flux $\phi=\pi/2$ for the ground state at half-filling.}
    \label{fig:flux_bowtie}
\end{figure}

In 2D, let us begin with a kind of very simple lattice, namely a bowtie lattice consisting of five sites as illustrated in Figure~\ref{fig:bowtie}.
As a generalization of the GRP on 1D rings, we have such a following corollary: 
\begin{corollary}
    For a half-filled repulsive Hubbard model with a minimal $|S_{\text{tot}}^{z}|$ defined on a bowtie as like Figure~\ref{fig:bowtie}, at any finite temperature, the optimal fluxes for the free energy $F$ are $\pm\pi/2$ in each triangle.
\end{corollary}
\begin{proof}
Carry on the GRP along the dashed line as illustrated in Figure~\ref{fig:bowtie}.
The other procedures are the same with the 1D ring.
\end{proof}
However, the GRP cannot tell us the sign of the fluxes in each triangle.
Due to the fact that the GRP only requires half-filled condition while does not care about whether $U=0.0$ or not, thus we can deduce the flux pattern for interacting fermions from free fermions at half-filling.
In another word, at half-filling, the (G)RP make a continuous connection between Hubbard interacting fermions and free fermions.
The optimal flux patterns are the same on both sides.

The optimal flux for free fermions on a bowtie at half-filling turns out sharing the identical sign of $\pm\pi/2$.
We verified this conclusion for interacting fermions by the numerical ED as shown in Figure~\ref{fig:flux_bowtie} and provide some more results deviating half-filling as a comparison.
We can see that on a 2D bowtie lattice, half-filling is still such a special filling that the free energy and ground state energy share the same optimal flux, which is immobile by changing the Hubbard interactions.
Otherwise not when it is away from half-filling.

\subsection{Triangular and Kagome lattices}
Now we turn to the flux problem on more complicated 2D non-bipartite lattices such as triangle and Kagome~\cite{JPSJ.83.034712, zengKagome, PhysRevB.90.081105} in terms of free as well as Hubbard interacting fermions.

\subsubsection{Half-filled free fermions coupled to a $\mathbbm{Z}_{4}$ gauge field}
We would like to consider the flux problem of half-filled free fermions both on the triangular and Kagome lattices in the first place.
According to our previous discussions, here $\pm\pi/2$ fluxes are also strongly implied in the plaquettes encircled by odd-numbered links.
Therefore, we would like to imagine that there is only a minimal $\mathbbm{Z}_{4}\subset \text{U}(1)$ gauge field coupled to these lattice fermions.
$\mathbbm{Z}_{4}$ is the smallest gauge group that possibly supports $\pm\pi/2$ fluxes.
A pure $\mathbbm{Z}_{4}$ lattice gauge theory should be taken the form as
\begin{equation}
    H_{g}=-\frac{J}{2}\sum_{p}\prod_{l\in p}\tau_{l}+h.c.,
\end{equation}
where $\tau_{l}\in G=\{1, \text{i}, -1, -\text{i}\}$ is the gauge connection living on the link $l$ of each plaquette $p$.
$G$ is our gauge group.
$J$ is the coupling constant.
Generically, this $\mathbbm{Z}_{4}$ gauge field can be coupled to lattice fermions through
\begin{equation}
    H_{f}=-\sum_{\sigma}\sum_{\langle{ij}\rangle}(\tau_{ij}c_{i\sigma}^{\dagger}c_{j\sigma}+h.c.)+H_{g}+H_{U}.
    \label{eq:ham_fermion_gauge}
\end{equation}
Here for the free fermions with $U=0.0$, we only consider the gauge field as a background, which means $J=0.0$.
The gauge field does not provide its own dynamics.
Given Eq.~(\ref{eq:ham_fermion_gauge}), exploring its whole phase diagram as well as phase transitions would be a quite interesting task, which, however, is deviating from the goal of this paper and left in the future.

Since the two branches $\sigma=\uparrow, \downarrow$ are decoupled and symmetric if we are on a lattice with even number of sites, we can only consider one of them and the Hamiltonian can be written as $H=\eta^{\dagger}\mathcal{H}\eta$ where $\eta=(c_{0}, c_{1}, \cdots, c_{N-1})^{T}$ assuming there are $N$ sites consisting of the lattice.
We should also assume that the magnetic unit cell can be only enlarged as much as $2\times 2$ larger than the original lattice unit cell.
The question remained is to find out which flux pattern is optimal to make this fermion-gauge field coupled system possessing the lowest ground state energy.

\begin{figure}[!ht]
    \centering
        \tikzset{
            hexa/.style={shape=regular polygon, regular polygon sides=6, minimum size=2, line width=0.5pt, draw, anchor=south, shape border rotate=90}
        }
        \begin{tikzpicture}[x=0.75cm, y=0.75cm]
        \tikzset{
              arrow/.style = {
                decoration={
                    markings, 
                    mark=at position 0.6 with{\arrow[black]{latex}}
                },
                postaction={decorate}
            }
        }
        \tikzset{
              arrow1/.style = {
                decoration={
                    markings, 
                    mark=at position 0.6 with{\arrow[red]{latex}}
                },
                postaction={decorate}
            }
        }
        
        \begin{scope}[shift={(0.0, 0.0)}]
        \def\r{0.1};
        \def\l{1.5};
        \pgfmathsetmacro\de{0.0};
        \draw[arrow, color=black, line width=1pt] (0.0, 0.0) -> (1.0*\l, 0.0);
        \draw[arrow, color=black, line width=1pt] (1.0*\l, 0.0) -- (2.0*\l, 0.0);
        \begin{scope}[shift={({2.0*\l*cos(60)}, {2.0*\l*sin(60)})}]
        \draw[arrow, color=black, line width=1pt] (0.0, 0.0) -> (1.0*\l, 0.0);
        \draw[arrow, color=black, line width=1pt] (1.0*\l, 0.0) -- (2.0*\l, 0.0);
        \end{scope}
        
        \draw[arrow, color=black, line width=1pt] (0.0, 0.0) -> ({1.0*\l*cos(60)}, {1.0*\l*sin(60)});
        \draw[arrow, color=black, line width=1pt] ({1.0*\l*cos(60)}, {1.0*\l*sin(60)}) -> ({2.0*\l*cos(60)}, {2.0*\l*sin(60)});
        \begin{scope}[shift={(2.0*\l, 0.0)}]
        \draw[arrow, color=black, line width=1pt] (0.0, 0.0) -> ({1.0*\l*cos(60)}, {1.0*\l*sin(60)});
        \draw[arrow, color=black, line width=1pt] ({1.0*\l*cos(60)}, {1.0*\l*sin(60)}) -> ({2.0*\l*cos(60)}, {2.0*\l*sin(60)});
        \end{scope}
        
        \begin{scope}[shift={({1.0*\l*cos(60)}, {1.0*\l*sin(60)})}]
        \draw[arrow, color=black, line width=1pt] (0.0, 0.0) -> (1.0*\l, 0.0);
        \draw[color=black, line width=1pt] (1.0*\l, 0.0) -- (2.0*\l, 0.0);
        \end{scope}
        
        \begin{scope}[shift={(1.0*\l, 0.0)}]
        \draw[color=black, line width=1pt] (0.0, 0.0) -> ({1.0*\l*cos(60)}, {1.0*\l*sin(60)});
        \draw[color=black, line width=1pt] ({1.0*\l*cos(60)}, {1.0*\l*sin(60)}) -> ({2.0*\l*cos(60)}, {2.0*\l*sin(60)});
        \end{scope}
        
        \draw[color=black, line width=1pt] ({1.0*\l*cos(60)}, {1.0*\l*sin(60)}) -- (1.0*\l, 0.0);
        \draw[color=black, line width=1pt] (2.0*\l, {2.0*\l*sin(60)}) -- ({2.0*\l+1.0*\l*cos(60)}, {1.0*\l*sin(60)});
        \draw[color=black, line width=1pt] (1.0*\l, {2.0*\l*sin(60)}) -- ({2.0*\l}, 0.0);
        
        \foreach \i in {0, 1, 2} {
            \pgfmathsetmacro\shift{1.0*\i*\l};
            \draw[fill, color=blue] ({0.0+\shift}, {0.0}) circle (\r);
        }
        \begin{scope}[shift={({1.0*\l*cos(60)}, {1.0*\l*sin(60)})}]
        \foreach \i in {0, 1, 2} {
            \pgfmathsetmacro\shift{1.0*\i*\l};
            \draw[fill, color=blue] ({0.0+\shift}, {0.0}) circle (\r);
        }
        \end{scope}
        \begin{scope}[shift={({2.0*\l*cos(60)}, {2.0*\l*sin(60)})}]
        \foreach \i in {0, 1, 2} {
            \pgfmathsetmacro\shift{1.0*\i*\l};
            \draw[fill, color=blue] ({0.0+\shift}, {0.0}) circle (\r);
        }
        \end{scope}
        
        \node[anchor=north] at ({0.5*\l}, {0.6*\l*sin(60)}) {$x$};
        \node[anchor=north] at ({1.0*\l}, {0.75*\l*sin(60)}) {$y$};
        \node[anchor=north] at ({\l}, -1.0) {(a)};
        \end{scope}
        
        \begin{scope}[shift={(5.0, 0.0)}]
        \draw[fill=gray!30, draw=gray!30] (0.5, {0.5*tan(30)}) ellipse (0.8 and 0.8);
        \def\r{0.1};
        \draw[fill, color=blue] (0.0, 0.0) circle (\r);
        \foreach \i in {0, 1, 2, 3} {
            \pgfmathsetmacro\shift{1.0*\i};
            \draw[arrow, color=black, line width=1pt] ({0.0+\shift}, 0.0) -- ({1.0+\shift}, 0.0);
        }
        \begin{scope}[shift={({2.0*cos(60)}, {2.0*sin(60)})}]
            \foreach \i in {0, 1, 2, 3} {
            \pgfmathsetmacro\shift{1.0*\i};
            \ifthenelse{\i<3}{\draw[arrow, color=black, line width=1pt] ({0.0+\shift}, 0.0) -- ({1.0+\shift}, 0.0);}{\draw[color=black, line width=1pt] ({0.0+\shift}, 0.0) -- ({1.0+\shift}, 0.0);}
        }
        \end{scope}
        \begin{scope}[shift={({4.0*cos(60)}, {4.0*sin(60)})}]
            \foreach \i in {0, 1, 2, 3} {
            \pgfmathsetmacro\shift{1.0*\i};
            \draw[arrow, color=black, line width=1pt] ({0.0+\shift}, 0.0) -- ({1.0+\shift}, 0.0);
        }
        \end{scope}
        
        \foreach \i in {0, 1, 2, 3} {
            \pgfmathsetmacro\xshift{1.0*\i*cos(60)};
            \pgfmathsetmacro\yshift{1.0*\i*sin(60)};
            \draw[arrow, color=black, line width=1pt] ({0.0+\xshift}, {0.0+\yshift}) -- ({1.0*cos(60)+\xshift}, {1.0*sin(60)+\yshift});
        }
        \begin{scope}[shift={(2.0, 0.0)}]
        \foreach \i in {0, 1, 2, 3} {
            \pgfmathsetmacro\xshift{1.0*\i*cos(60)};
            \pgfmathsetmacro\yshift{1.0*\i*sin(60)};
            \ifthenelse{\i=0 \OR \i=2}{\draw[arrow, color=black, line width=1pt] ({0.0+\xshift}, {0.0+\yshift}) -- ({1.0*cos(60)+\xshift}, {1.0*sin(60)+\yshift});}{\draw[color=black, line width=1pt] ({0.0+\xshift}, {0.0+\yshift}) -- ({1.0*cos(60)+\xshift}, {1.0*sin(60)+\yshift});}
        }
        \end{scope}
        \begin{scope}[shift={(4.0, 0.0)}]
        \foreach \i in {0, 1, 2, 3} {
            \pgfmathsetmacro\xshift{1.0*\i*cos(60)};
            \pgfmathsetmacro\yshift{1.0*\i*sin(60)};
            \draw[arrow, color=black, line width=1pt] ({0.0+\xshift}, {0.0+\yshift}) -- ({1.0*cos(60)+\xshift}, {1.0*sin(60)+\yshift});
        }
        \end{scope}
        
        \draw[color=black, line width=1pt]({1.0}, 0.0) -- ({1.0*cos(60)}, {1.0*sin(60)});
        \draw[color=black, line width=1pt]({1.0+3.0+3.0*cos(60)}, {0.0+3.0*sin(60)}) -- ({1.0*cos(60)+3.0+3.0*cos(60)}, {1.0*sin(60)+3.0*sin(60)});
        \draw[color=black, line width=1pt]({3.0}, {0.0}) -- ({3.0*cos(60)}, {3.0*sin(60)});
        \draw[color=black, line width=1pt]({3.0+1.0+1.0*cos(60)}, {0.0+1.0*sin(60)}) -- ({3.0*cos(60)+1.0+1.0*cos(60)}, {3.0*sin(60)+1.0*sin(60)});
        
        \foreach \i in {0, 1, 2, 3, 4} {
            \pgfmathsetmacro\shift{1.0*\i};
            \draw[fill, color=blue] ({0.0+\shift}, {0.0}) circle (\r);
        }
        \begin{scope}[shift={({2.0*cos(60)}, {2.0*sin(60)})}]
        \foreach \i in {0, 1, 2, 3, 4} {
            \pgfmathsetmacro\shift{1.0*\i};
            \draw[fill, color=blue] ({0.0+\shift}, {0.0}) circle (\r);
        }
        \end{scope}
        \begin{scope}[shift={({4.0*cos(60)}, {4.0*sin(60)})}]
        \foreach \i in {0, 1, 2, 3, 4} {
            \pgfmathsetmacro\shift{1.0*\i};
            \draw[fill, color=blue] ({0.0+\shift}, {0.0}) circle (\r);
        }
        \end{scope}
        
        \foreach \i in {1, 3} {
            \pgfmathsetmacro\xshift{1.0*\i*cos(60)};
            \pgfmathsetmacro\yshift{1.0*\i*sin(60)};
            \draw[fill, color=blue] ({0.0+\xshift}, {0.0+\yshift}) circle (\r);
        }
        \begin{scope}[shift={(2.0, 0.0)}]
        \foreach \i in {1, 3} {
            \pgfmathsetmacro\xshift{1.0*\i*cos(60)};
            \pgfmathsetmacro\yshift{1.0*\i*sin(60)};
            \draw[fill, color=blue] ({0.0+\xshift}, {0.0+\yshift}) circle (\r);
        }
        \end{scope}
        \begin{scope}[shift={(4.0, 0.0)}]
        \foreach \i in {1, 3} {
            \pgfmathsetmacro\xshift{1.0*\i*cos(60)};
            \pgfmathsetmacro\yshift{1.0*\i*sin(60)};
            \draw[fill, color=blue] ({0.0+\xshift}, {0.0+\yshift}) circle (\r);
        }
        \end{scope}
        
        \node[anchor=south] at (0.5, {0.1*sin(60)}) {$x$};
        \node[anchor=north] at (2.5, {1.9*sin(60)}) {$y$};
        \node[anchor=north] at (1.5, {1.2*sin(60)}) {$z$};
        \node[anchor=north] at (2.0, -1.0) {(b)};
    \end{scope}
\end{tikzpicture}
\caption{$2\times 2$ enlarged magnetic unit cell with respect to fixed gauge choices. An arrow denotes a gauge connection $\exp(+\text{i}\pi/2)$ fixed on this very link. (a) Triangular lattice with a flux state labelled by $[x, y]$. (b) Kagome lattice with a flux state labelled by $[x, y, z]$.}
   \label{fig:freeTriangleKagome}
\end{figure}
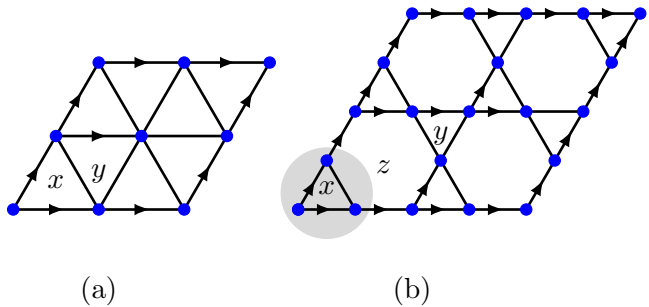

For the triangular lattice as shown in Figure~\ref{fig:freeTriangleKagome}(a), within the enlarged magnetic unit cell, we have a specific fixed gauge choice as the arrows indicate then leave the other $7$ links free.
The undetermined gauge field connections without arrows can be chosen from the $\mathbbm{Z}_{4}$ gauge group at liberty.
In this sense, there are $\mathcal{N}=4^{7}=16,384$ kinds of choices in total.
We searched all these possibilities numerically on a $16\times 16$ triangular lattice with periodic boundary conditions (PBC).
$\pm[\pi/2, \pi/2]$ turn out to be the optimal flux patterns.
For the Kagome lattice as shown in Figure~\ref{fig:freeTriangleKagome}(b), within the enlarged magnetic unit cell, there are $11$ free links determining various flux patterns.
We searched all the $\mathcal{N}=4^{11}=4,194,304$ kinds of possibilities numerically on a $4\times 4\times 3$ Kagome lattice with PBC.
$\pm[\pi/2, \pi/2, 0]$ are the two optimal flux patterns for the Kagome lattice.
Some representative flux states and their energy are enumerated in TABLE.~\ref{tab:freeTriangleKagome}.
In Figure~\ref{fig:finite_size} we also show a finite-size effect of the energy for these flux states.
On both triangular and Kagome lattices,  we can see that every rhomboid still prefers a $\pi$ flux while a $0$ flux gives much higher energy.
Moreover, these optimal flux states converge more quickly.
It seems that they are much less sensitive to the finite-size effect in comparison with other flux states.

\begin{figure}[!ht]
    \centering
    \includegraphics[width=0.85\textwidth]{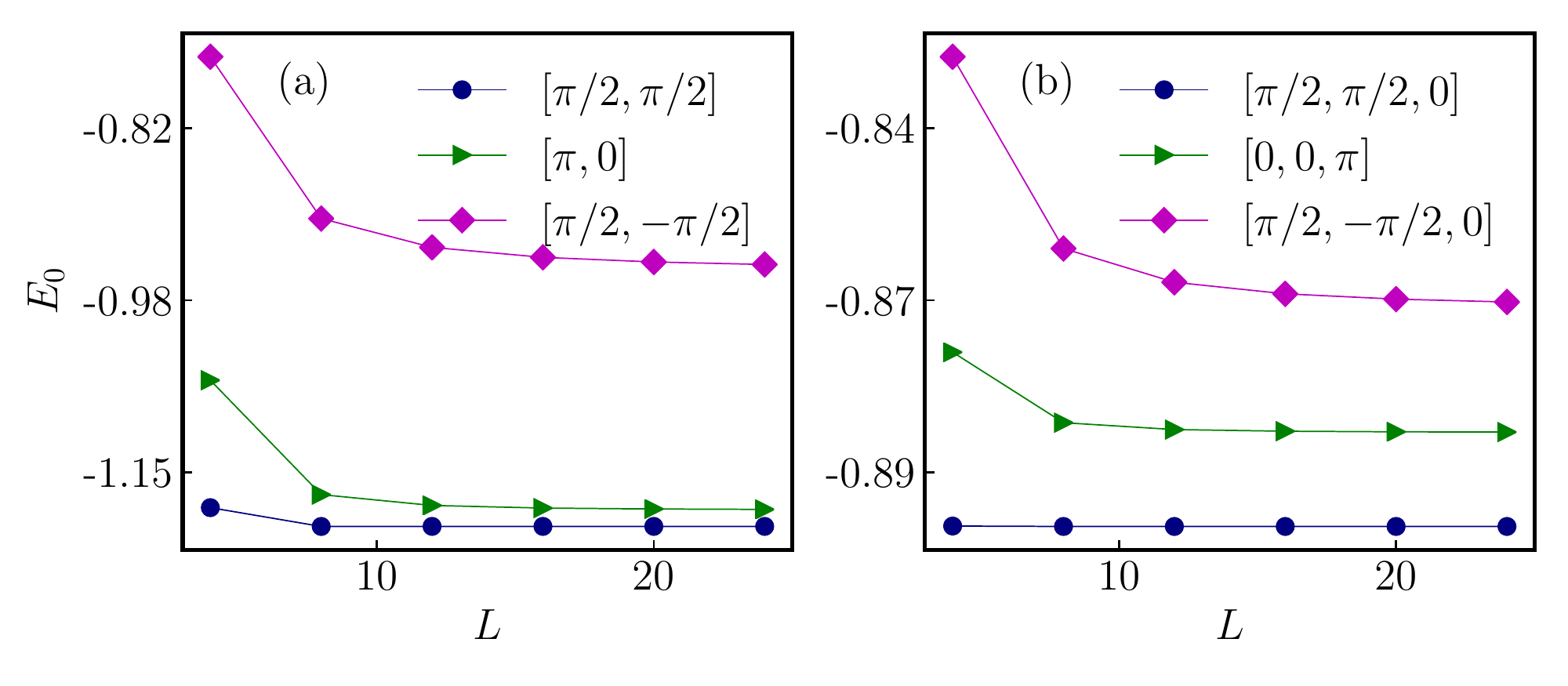}
    \caption{Finite-size scaling of the ground state energy for kinds of flux states on $L\times L$ lattices with PBC. (a) Triangular lattice. (b) Kagome lattice.}
    \label{fig:finite_size}
\end{figure}

\begin{table}
\caption{\label{tab:freeTriangleKagome} Per-site energy of some flux states for half-filled free fermions on a $16\times 16$ triangular lattice and a $8\times 8\times 3$ Kagome lattice with PBC}
\begin{tabular}{p{.2\textwidth}p{.2\textwidth}p{.2\textwidth}p{.2\textwidth}}
\hline
\hline
\multicolumn{2}{c}{Triangle} & \multicolumn{2}{c}{Kagome} \\
\hline
  $[+\pi/2, +\pi/2]$ & $-1.20102$ & $[+\pi/2, +\pi/2, 0]$ & -0.89912 \\
  $[\pi, 0]$ & $-1.18341$ & $[0, 0, \pi]$ & -0.88413 \\
  $[+\pi/2, -\pi/2]$ & $-0.94263$ & $[+\pi/2, -\pi/2, 0]$ & -0.85897 \\
\hline
\hline
\end{tabular}
\end{table}

\subsubsection{Turn on Hubbard interactions}
\begin{figure}[!ht]
    \centering
        \tikzset{
            hexa/.style={shape=regular polygon, regular polygon sides=6, minimum size=2, line width=0.5pt, draw, anchor=south, shape border rotate=90}
        }
        \begin{tikzpicture}[x=0.75cm, y=0.75cm]
        \tikzset{
            arrow/.style = {
                decoration={
                    markings, 
                    mark=at position 0.6 with{\arrow{latex}}
                },
                postaction={decorate}
            }
        }

        \begin{scope}[shift={(0.0, 0.0)}]
        \def\r{0.1};
        \def\l{1.25};
        \pgfmathsetmacro\de{0.0};
        \draw[dashed, line width=0.5pt] ({0.5*sin(60)*\l}, -0.75) -- ({0.5*sin(60)*\l}, {3.0*\l});
        \draw[color=black, line width=1pt] (0.0, 0.0) -- (0.0, 2.0*\l) -- ({sin(60)*\l}, {2.5*\l}) -- ({sin(60)*\l}, {0.5*\l}) -- (0.0, 0.0);
        \draw[color=black, line width=1pt] (0.0, {1.0*\l}) -- ({sin(60)*\l}, {0.5*\l});
        \draw[color=black, line width=1pt] (0.0, {2.0*\l}) -- ({sin(60)*\l}, {1.5*\l});
        \draw[color=black, line width=1pt] (0.0, {1.0*\l}) -- ({sin(60)*\l}, {1.5*\l});
        \draw[fill, color=blue] (0.0, 0.0) circle (\r);
        \draw[fill, color=blue] (0.0, {1.0*\l}) circle (\r);
        \draw[fill, color=blue] (0.0, {2.0*\l}) circle (\r);
        \draw[fill, color=blue] ({sin(60)*\l}, {2.5*\l}) circle (\r);
        \draw[fill, color=blue] ({sin(60)*\l}, {1.5*\l}) circle (\r);
        \draw[fill, color=blue] ({sin(60)*\l}, {0.5*\l}) circle (\r);
        \node[anchor=north] at ({0.5*sin(60)*\l}, -1.0) {(a)};
        \end{scope}
 
        \begin{scope}[shift={(4.0, 0.0)}]
        \draw[dashed, line width=0.5pt] (2.25, {-0.5*sqrt(3.0)}) -- (2.25, {2.0*sqrt(3.0)});
        \draw[fill=gray!30, draw=gray!30] (0.5, {0.5*tan(30)}) ellipse (0.8 and 0.8);
        \def\r{0.1};
        \foreach \j in {0, 1} { 
            \foreach \i in {0, 1} {
                \pgfmathsetmacro\xshift{2.0*\i+\j};
                \pgfmathsetmacro\yshift{sqrt(3.0)*\j};
                \draw[line width=1pt] (0.0+\xshift, 0.0+\yshift) -- (1.0+\xshift, 0.0+\yshift);
                \draw[line width=1pt] (1.0+\xshift, 0.0+\yshift) -- (0.5+\xshift, {0.5*sqrt(3.0)+\yshift});
                \draw[line width=1pt] (0.5+\xshift, {0.5*sqrt(3.0)+\yshift}) -- (0.0+\xshift, 0.0+\yshift);

                \ifthenelse{\i<1}{\draw[line width=1pt] (1.0+\xshift, 0.0+\yshift) -- (2.0+\xshift, 0.0+\yshift);}{}
                \ifthenelse{\j<1}{\draw[line width=1pt] (0.5+\xshift, {0.5*sqrt(3.0)+\yshift}) -- (1.0+\xshift, {sqrt(3.0)+\yshift});}{}
                \ifthenelse{\i>0}{\ifthenelse{\j<1}{\draw[line width=1pt] (0.5+\xshift, {0.5*sqrt(3.0)+\yshift}) -- (0.0+\xshift, {sqrt(3.0)+\yshift});}{}}{}

                \draw[fill, color=blue] (0.0+\xshift, 0.0+\yshift) circle (\r); 
                \draw[fill, color=blue] (1.0+\xshift, 0.0+\yshift) circle (\r);
                \draw[fill, color=blue] (0.5+\xshift, {0.5*sqrt(3.0)+\yshift}) circle (\r);
            }
        }
        \node[anchor=north] at (2.0, -1.0) {(b)};
    \end{scope}
\end{tikzpicture}
\caption{Carrying out the GRP on a (a) triangular lattice and (b) a Kagome lattice.
Within a lattice unit cell (shadowed area) of Kagome lattice, there are three inequivalent sites.}
   \label{fig:bowtieKagome}
\end{figure}
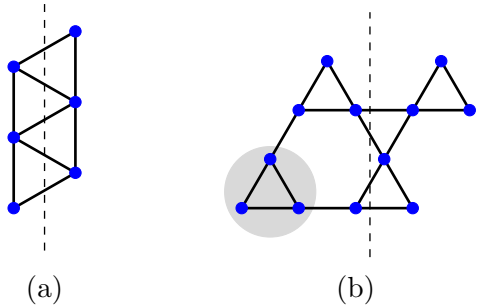
As a generalization of 1D rings and 2D bowtie lattice, we can carry out the GRP along the dashed lines on these kinds of lattices as shown in Figure~\ref{fig:bowtieKagome}.
Note that the direct application of our GRP is that, every nonintersected link will acquire a pure imaginary gauge connection in order to maximize the partition function.
In this sense, we can have the two following corollaries:
\begin{corollary}
    For a half-filled repulsive Hubbard model with a minimal $|S_{\text{tot}}^{z}|$ defined on a triangular lattice, the optimal fluxes for its free energy $F$ at any finite temperature are $\pm\pi/2$ in each triangle. 
\end{corollary}
\begin{corollary}
    For a half-filled repulsive Hubbard model with a minimal $|S_{\text{tot}}^{z}|$ defined on the Kagome lattice, the optimal flux patterns for its free energy $F$ at any finite temperature are $\pm\pi/2$ in each triangle and $0$ or $\pi$ in each hexagon.
\end{corollary}
On one hand, if only based on the GRP, analytically we still cannot nail down the sign of the triangle flux $\pm\pi/2$ for the half-filled Hubbard interacting fermions on the triangular and Kagome lattices.
Although E. H. Lieb's original result~\cite{lieb1993} implies that, every rhomboid of a triangular lattice should prefer a $\pi$ flux rather than $0$, the RP cannot be directly applied here since we have an effective next nearest neighbor hopping for each rhomboid.
As we have mentioned, at the particular half-filling, the (G)RP can make a bridge between free and Hubbard interacting fermions.
It doesn't care whether Hubbard $U$ is zero or not.
On the other hand, our numerical study of free fermions coupled to a static $\mathbbm{Z}_{4}$ gauge field on the triangular and Kagome lattices has shown the optimal flux patterns, but we still cannot state them as proof since we have imposed the magnetic unit cell assumption as well as a finite lattice size.
Combined these two hands, at least we can strongly believe that the optimal flux patterns for repulsive Hubbard model at half-filling are $\pm[\pi/2, \pi/2]$ and $\pm[\pi/2, \pi/2, 0]$ on the triangular and Kagome lattices, respectively.

For the triangular lattice, we have another argument:
we can carry out the ordinary RP along the dashed line as illustrated in Figure~\ref{fig:bowtieKagome} with a \emph{glide geometric reflection} $\widetilde{\mathcal{R}}$ instead of a direct geometric reflection $\mathcal{R}$, saying that the every rhomboid of the triangular lattice prefers a $\pi$ flux, which coincides with the optimal flux patterns $\pm[\pi/2, \pi/2]$.

\subsection{Numerical verification and flux singularities induced by strong Hubbard interactions}
\begin{figure}[!ht]
    \centering
    \includegraphics[width=0.85\textwidth]{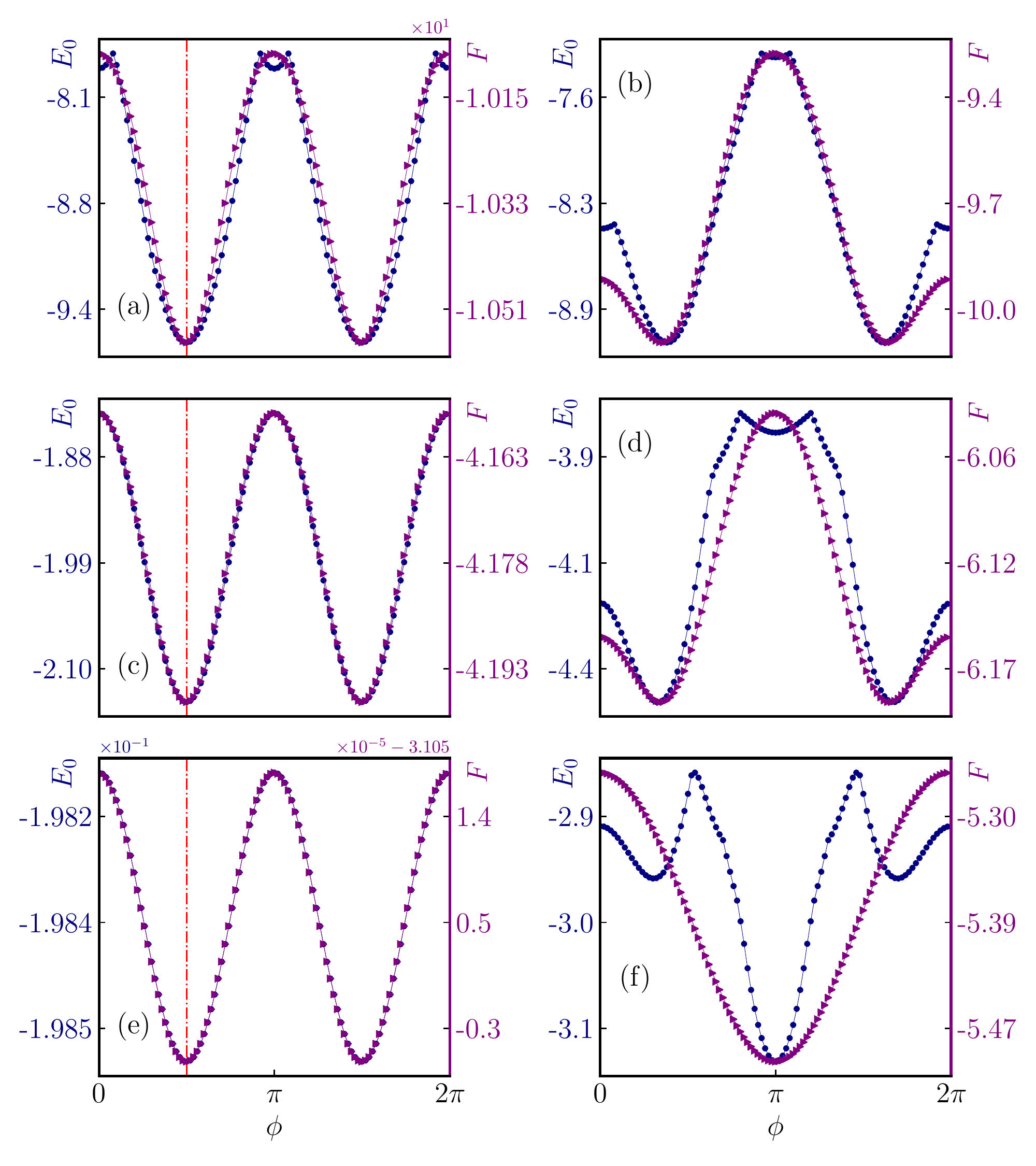}
    \caption{Ground state energy $E_{0}$ and finite temperature free energy $F$ of the Hubbard model on a triangular lattice consisting of $6$ sites. Vertically, (a, b), (c, d) and (e, f) denote $U/t=0.0, 10.0, 100.0$, respectively. Horizontally, (a, c, e) denote half-filling $N_{\uparrow}=N_{\downarrow}=3$. (b, d, f) denote filling $N_{\uparrow}=3, N_{\downarrow}=2$. Free energy is computed at $\beta=1.0$. The red dashed line marks the optimal flux $\phi=\pi/2$ for the ground state at half-filling.}
    \label{fig:flux_triangle}
\end{figure}
\begin{figure}[!ht]
    \centering
    \includegraphics[width=0.85\textwidth]{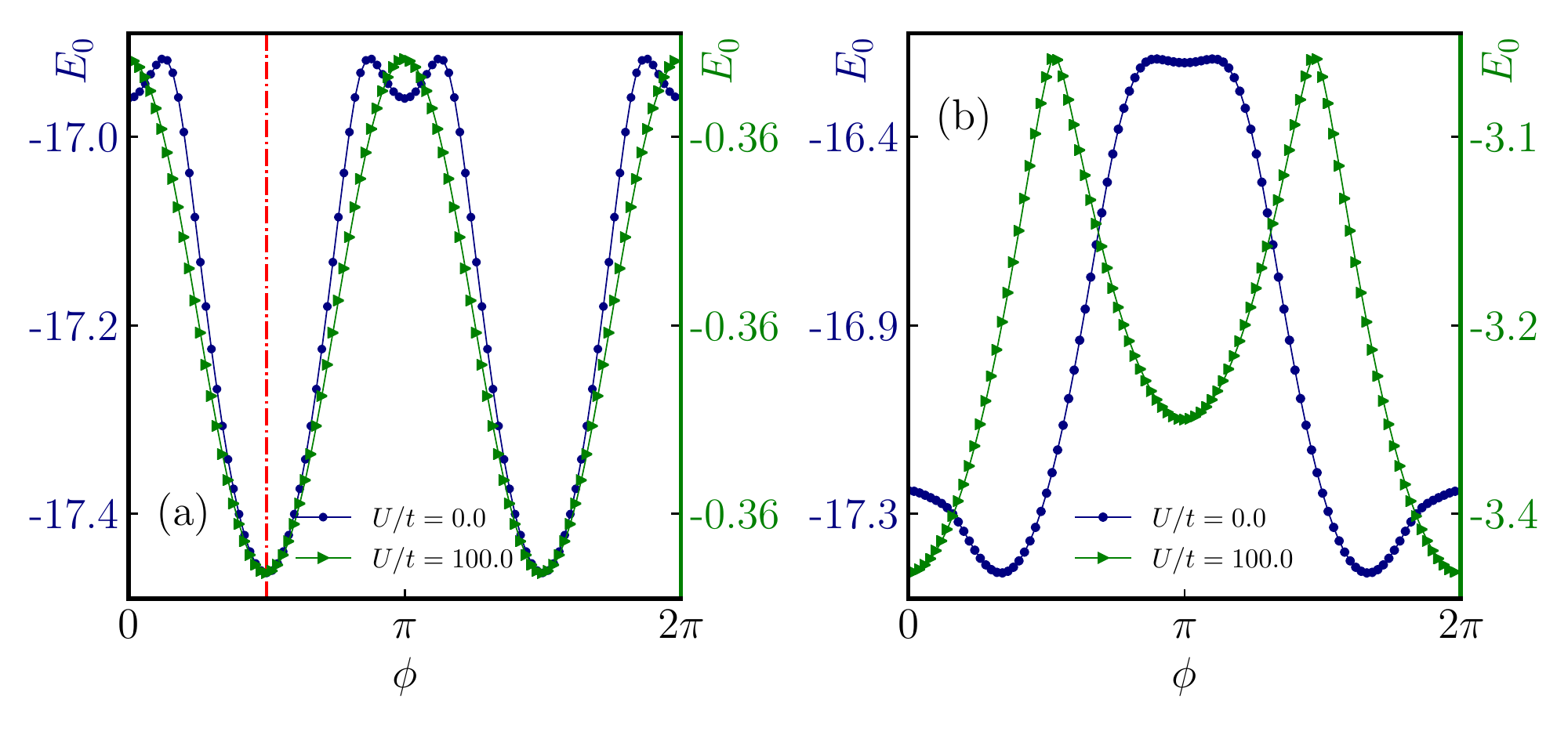}
    \caption{Ground state energy $E_{0}$ of the Hubbard model on a $2\times 2\times 3$ Kagome lattice with open boundary condition.
    Flux $\phi$ is added through each triangle while there is no flux threading through the hexagon. (a) refers to the half-filled $N_{\uparrow}=6, N_{\downarrow}=6$. The dashed line denotes $\pi/2$. (b) refers to the Nagaoka filling $N_{\uparrow}=6, N_{\downarrow}=5$.}
    \label{fig:flux_kagome}
\end{figure}
Meanwhile, we simulate the interacting fermions by the numerical ED on both triangular and Kagome lattices as shown in Figures~\ref{fig:flux_triangle} and~\ref{fig:flux_kagome}.
Note that we cannot obtain the free energy for the Hubbard model up to $N=12$ sites by ED but only the ground state energy because of the iterative algorithm~\cite{arpackpp}.

On the Kagome lattice, the half-filled case shown in Figure~\ref{fig:flux_kagome}(a) is consistent with our previous discussion.
The optimal flux pattern even up to $U/t=100.0$ is still as same as the free fermions.
For the Nagaoka filling, we could see that, in Figure~\ref{fig:flux_kagome}(b), Hubbard interactions can dramatically change the optimal flux pattern.
The optimal flux for free fermions locates at $\phi\simeq \pi/3$ there, while with $U/t=100.0$, the ground state with minimal energy locates at $\phi=0$.

For the triangular lattice, by ED we also numerically checked that the energy of flux state $[\pi/2, \pi/2]$ is lower than the state $[0, \pi]$~\cite{PhysRevLett.114.167201}, which are both $\pi$-flux state in a rhomboid.

For the bowtie lattice, when it is away from half-filling as shown in Figure~\ref{fig:flux_bowtie}(b, d, f), there is no flux singularity for free fermions.
But sufficiently strong Hubbard interactions can give birth to new emergent flux singularities as shown in Figure~\ref{fig:flux_bowtie}(f).
As discussed in 1D, we think that it is still quite reasonable to regard these very singularities as some emergent Luttinger-like NFLs features driven by strong interactions in 2D~\cite{Varma2002, RevModPhys.73.797, 10.1146}.
On the Kagome lattice with a Nagaoka filling in Figure~\ref{fig:flux_kagome}(b), we have also numerically observed similar emergent flux singularities for $U/t=100.0$.
They are the Luttinger-like evidences in pure 2D, which contradict with Fermi-liquid theory.

\section{Conclusion and discussion}
\label{sec:summary}
\subsection{Brief summary}
In this paper, we did a systematical study on the flux problem of free as well as Hubbard interacting fermions on non-bipartite lattices.
At half-filling, several rigorous results can be proved by utilizing the GRP technique.
Other fillings are studied numerically.
On a 1D non-bipartite odd-numbered ring at half-filling, the optimal flux for the ground state is at an unrenormalized $\pm\pi/2$.
While away from half-filling, the optimal flux can be altered by strong Hubbard interactions.
In 2D at half-filling, the optimal flux of the Hubbard model on the triangular lattice, $\pm[\pi/2, \pi/2]$ states are ascertained as the lowest energy flux states.
For the Kagome lattice, they are $\pm[\pi/2, \pi/2, 0]$.
On both triangular and Kagome lattice, the chiral order parameter for each triangle, namely $\mathbf{\sigma}_{1}\cdot(\mathbf{\sigma}_{2}\times\mathbf{\sigma}_{3})$, is not only nonvanishing but also maximized in the lowest energy optimal flux states.
We also addressed and discussed the numerically observed emergent flux singularities driven by strong Hubbard interactions and attributed them to some Luttinger-like NFL features.
It is easy and straightforward to generalize our results to the extended Hubbard model~\cite{PhysRevLett.73.2158, PhysRevLett.92.236401} while we would not expatiate on it here.
For other more complicated non-bipartite lattices such as decorated honeycomb~\cite{PhysRevA.90.053627} and 3D lattices~\cite{PhysRevLett.101.197202, PhysRevLett.80.2933}, we still have similar conclusions if we are in the same fermion-gauge field scenario, meaning that $\pm\pi/2$ fluxes will possibly emerge from the plaquettes with odd-numbered sites.

So far we can provide a basic answer to the question asked in the Introduction Sec.~\ref{sec:intro}:
If we lose bipartiteness, for the half-filled free as well as Hubbard interacting fermions coupled to appropriate gauge fields, time-reversal symmetry $\mathcal{T}$ and parity symmetry $\mathcal{P}$ break spontaneously.
Chiral fermions will emerge.
Ground states are not unique any longer.
The sign problem cannot be avoided if simulated by quantum Monte Carlo.

\subsection{Possible implication to quantum spin liquids}
As we have mentioned, the fermion-gauge theory coupled scheme is not only a fantastic scenario but to be quite realistic.
A pure bosonic model can be written in terms of slave-fermions~\cite{PhysRevB.37.3774, PhysRevB.49.5200}.
Gauge fields indeed can emerge from these strongly correlated bosonic quantum systems.

In recent years, the Kagome lattice has attracted a great deal of attention with respect to the study of QSLs~\cite{PhysRevLett.98.117205, PhysRevB.83.224413, PhysRevB.84.020407, PhysRevLett.112.137202, Bauer2014, PhysRevB.91.041124, PhysRevX.7.031020, Zhu5535}.
On the projected mean-field level of slave-fermions, the Dirac state $[0, 0, \pi]$ is always reported to serve as the parent state for various kinds of QSLs on the Kagome lattice~\cite{PhysRevLett.98.117205}.
However, in this paper we found that, for not only free but also Hubbard interacting fermions at half-filling, the ground state energy of the flux states $\pm[\pi/2, \pi/2, 0]$ is lower than the Dirac state $[0, 0, \pi]$.
Although the ground state wavefunction of the Hubbard model is not equivalent to Gutzwiller projected mean-field wavefunctions, we think it is still meaningful to study QSLs starting from the chiral $\pm[\pi/2, \pi/2, 0]$ states.
As well as on the triangular lattice, these optimal flux states break time-reversal symmetry $\mathcal{T}$. 
Even up to very strong Hubbard interactions, we have demonstrated that the lowest energy ground state of these half-filled fermion-gauge theory coupled systems always tend to select some $\pm\pi/2$ fluxes within plaquettes enclosed by odd-numbered links.

In this sense, our results strongly imply that time-reversal symmetry breaking chiral fermions and chiral QSLs~\cite{PhysRevB.39.11413, PhysRevB.93.094437} may be widely present as well as stabilized by emergent gauge fields in strongly correlated bosonic quantum systems on various kinds of frustrated non-bipartite lattices.
Possible examples have been reported widely such as in the Ref.s~\cite{Taguchi2573, PhysRevLett.99.247203, Gong2014, PhysRevA.93.061601, PhysRevLett.117.167202, PhysRevX.10.021042, PhysRevB.100.241111, PhysRevB.87.205111}.
While note that, an emergent $\mathbbm{Z}_{2}$ gauge field can only allow a Dirac state which does not break time-reversal symmetry $\mathcal{T}$.
If $\pm\pi/2$ fluxes are expected, the emergent gauge field must at least be a $\mathbbm{Z}_{4}$ one.
$\text{U}(1)$ gauge field is also possible but it is more subtle because of the possible confinement~\cite{Polyakov1977, PhysRevB.72.235104, PhysRevX.9.021022}.

\subsection{Outlook}
So far we have studied lattice fermions on several kinds of non-bipartite lattices, especially at half-filling.
We hope to continue to explore the fruitful features of strongly interacting fermions on non-bipartite lattices away from half-filling.
For example, if one hole doped, the Nagaoka state emerges on bipartite lattices if Hubbard interactions are sufficiently strong.
However, as far as we know, we do not have much information for the Hubbard model with Nagaoka filling on non-bipartite lattices yet.

In the continuum limit, a field theoretical description is still needed.
In 1D, the conventional bosonization technique can be applied to the system, which might lead to some new perspectives and physics.
Furthermore, we also know that chiral spin states are deeply related to the topological Chern-Simons term in field theory, which breaks the parity symmetry $\mathcal{P}$ and time-reversal symmetry $\mathcal{T}$ simultaneously. Actually, fermion statistics can also be altered by the gauge fields~\cite{POLYAKOV1988, PhysRevX.6.031043, Murugan2017} thus bosons and even anyons may emerge.
Further discussion on the application of our current results towards quantum Hall effects~\cite{PhysRevLett.50.1395, PhysRevLett.59.2095} and topological quantum field theories hopefully would be addressed in the future.

\section*{Acknowledgement}
W. Z. would like to thank D.N. Sheng for her hospitality and helpful instructions at CSUN, California, where this work was inspired and initialized.
W. Z. thanks Y.M. Lu, T.L. Ho, X.Z. Feng, Z.Y. Weng and A. Rasmussen for many valuable discussions.
W. Z. thanks the Referee for bringing several helpful suggestions to strengthen the revised manuscript.
W. Z. also acknowledges the Unity, a well-managed high-performance computing cluster in the College of Arts and Sciences of the Ohio State University.
This work is supported by the National Science Foundation under award number DMR-1653769.
\newpage
\appendix

\section{1D free spinless fermions and the Jordan-Wigner transformation}
\label{sec:1d_spinless}
The Hamiltonian of spinless free fermions on such a lattice can be written as
\begin{equation}
    H_{0}=-\sum_{j=0}^{L-1}\left(t_{j, j+1} c_{j}^{\dagger}c_{j+1}+h.c. \right),
    \label{eq:ham_free_spinless}
\end{equation}
the Hamiltonian Eq. (\ref{eq:ham_free_spinless}) becomes $H_{0}=-\sum_{k}(2\cos{k})c_{k}^{\dagger}c_{k}$ if we set $t_{j,j+1}=1.0$.
Suppose there is a fixed number of $N (N\leqslant L)$ spinless fermions living on the ring and preserve the $\text{U}(1)$ symmetry.
Different boundary conditions will impose different quantization conditions of $k$.
For example, PBC gives $kL=2l\pi, l\in\mathbbm{Z}_{+}$.
If $N$ is even, there is a two-fold ground state degeneracy.
Anti-periodic boundary condition gives $kL=(2l+1)\pi, l\in\mathbbm{Z}_{+}$.
If $N$ is odd, there is a two-fold ground state degeneracy.

The ground state is simply the one in which the lowest $N$ orbitals are occupied. 
More generally suppose there is a flux $\phi\in[0, 2\pi)$ threading the ring, we have the shifted $k=(2l\pi+\phi)/L, l\in\mathbbm{Z}_{+}$.
When it comes to the ground state(s), for $N=2n+1, n\in\mathbbm{Z}_{+}$, the optimal flux is $0$ since the lowest mode with $k=0$ can be occupied by one fermion. 
For $N=2n, n\in\mathbbm{Z}_{+}$, the optimal flux is $\pi$ for the sake of symmetric band filling.

On the other hand, we know that 1D spinless fermions can be exactly mapped to a spin model through the Jordan-Wigner transformation~\cite{RevModPhys.51.659}.
For $j=0, c_{0}=\sigma_{0}^{-}, c_{0}^{\dagger}=\sigma_{0}^{+}$; and for $j\geqslant{1}$,
\begin{equation}
\begin{aligned}
    c_{j}&=\left[\prod_{l=0}^{j-1}e^{\text{i}\pi\sigma^{+}_{l}\sigma^{-}_{l}}\right]\sigma^{-}_{j}=\left[\prod_{l=0}^{j-1}\left(-\sigma^{z}_{l}\right)\right]\sigma^{-}_{j} \\
    c^{\dagger}_{j}&=\sigma^{+}_{j}\left[\prod_{l=0}^{j-1}e^{-\text{i}\pi\sigma^{+}_{l}\sigma^{-}_{l}}\right]=\sigma^{+}_{j}\left[\prod_{l=0}^{j-1}\left(-\sigma^{z}_{l}\right)\right],
\end{aligned}
\end{equation}
where $\sigma^{\pm}=\left(\sigma^{x}\pm\text{i}\sigma^{y}\right)/2$.
$\sigma^{x, y, z}$ are the Pauli matrices.
Specifically, a 1D spin-$1/2$ XY model reads
\begin{equation}
    \begin{aligned}
        H_{XY}&=-\sum_{j=0}^{L-1}\frac{1}{2}\left(t_{j,j+1}\sigma_{j}^{x}\sigma_{j+1}^{x}+\sigma_{j}^{y}\sigma_{j+1}^{y}\right) \\
        &=-\sum_{j}^{L-1}\left(t_{j,j+1}\sigma_{j}^{+}\sigma_{j+1}^{-}+h.c.\right),
    \end{aligned}
    \label{eq}
\end{equation}
which is equivalently to a hard-core boson model
\begin{equation}
    H_{b}=-\sum_{j}^{L-1}\left(t_{j,j+1}b_{j}^{\dagger}b_{j+1}+h.c.\right)
    \label{eq:ham_hard_core_boson}
\end{equation}
if we identify $\sigma^{+}\sim b^{\dagger}, \sigma^{-}\sim b$ as the creation and annihilation operators of the hard-core bosons.
Assuming $t_{j,j+1}$ are real and $t_{j,j+1}>0, \forall{j}$, these hard-core bosons are on a periodic ring without any flux inserted.
Note that, if we carry out the Jordan-Wigner transformation of these bosons, a boundary term will appear in the fermionic model like 
\begin{equation}
    H_{JW}=-\sum_{j=0}^{L-2}\left(t_{j,j+1}c_{j}^{\dagger}c_{j+1}+h.c.\right)+Q\left(t_{L-1,0}c_{L-1}^{\dagger}c_{0}+h.c.\right),
    \label{eq:ham_jw_free_fermion}
\end{equation}
where $Q=\prod_{j=0}^{L-1}\left(-\sigma_{j}^{z}\right)=(-)^{\sum_{j}c_{j}^{\dagger}c_{j}}=(-)^{N}$ is the \emph{parity operator}.
That is, if we require Eq.~(\ref{eq:ham_jw_free_fermion}) and Eq.~(\ref{eq:ham_hard_core_boson}) are exactly mapped to each other, the boundary condition for Eq.~(\ref{eq:ham_jw_free_fermion}) must be compatible with the parity of fermions.
If $N$ is even, $Q=+1$ denoting that the Jordan-Wigner fermions have an anti-periodic boundary condition (APBC).
In another word, a fictitious $\pi$-flux is effectively inserted through the ring.
If $N$ is odd, $Q=-1$ denoting a PBC, which means a $0$-flux is inserted through the ring.
According to the \emph{natural inequality}~\cite{PhysRevLett.111.100402, PhysRevB.97.125153} theorem, the ground state energy of these hard-core bosons would never be greater than spinless fermions if they share a same form of Hamiltonians like Eq.~(\ref{eq:ham_hard_core_boson}) and Eq.~(\ref{eq:ham_free_spinless}).
If the parity is compatible with the boundary condition meaning that the Jordan-Wigner transformation is exactly valid, Eq.~(\ref{eq:ham_hard_core_boson}) and Eq.~(\ref{eq:ham_jw_free_fermion}) will have the same ground state energies, which means the lower bound of the natural inequality is touched.
Thus we conclude that the optimal flux for spinless free fermions is $\pi$ or $0$ if the particle number is even or odd, respectively.
In a word, the optimal flux states for spinless fermions should be non-chiral, which preserve the time reversal symmetry.

In the case of finite temperature and free energy $F=-1/\beta\ln{Z}$ with inversed temperature $\beta$, we have
\begin{lemma}
    For free spinless fermions with the Hamiltonian defined by Eq.~(\ref{eq:ham_free_spinless}) on a ring lattice, at any finite temperature the optimal flux for free energy is identical to the optimal flux in its ground state, namely $0$ or $\pi$ depending on the parity of particle number $N$ is odd or even.
\end{lemma}
\begin{proof}
The canonical partition function reads
\begin{equation}
    Z=\text{tr}\left( e^{-\beta H_{0}} \right)=\lim_{M\rightarrow\infty}\text{tr}\left[ V^{M}(\phi) \right],
    \label{eq:}
\end{equation}
and
\begin{equation}
    V(\phi)=1+\delta\left( \sum_{j=0}^{L-2}c_{j}^{\dagger}c_{j+1}+e^{\text{i}\phi}c_{L-1}^{\dagger}c_{0}+h.c. \right),
    \label{eq:}
\end{equation}
where we have chosen a specific gauge and defined $\delta\equiv t\beta/M$.
For a fixed $M$, $V^{M}(\phi)$ can be expanded to a polynomial as $V^{M}(\phi)=\sum_{\alpha}X_{\alpha}$, where $X_{\alpha}$ is a string of fermionic operators.
First of all, we make a convention to label all the lattice sites in a 1D array and then write the basis of the Hilbert space as $|\alpha\rangle\equiv\{c_{}^{\dagger}c_{}^{\dagger}\cdots|0\rangle\}$ with a fixed order of lattice fermions.
The non-vanishing $\text{tr}\left( X \right)$ requires that $X$ must recover a basis configuration to itself.
In this sense, there are two kinds of operator strings: trivial ones such as $\mathbbm{1}$ and $c_{0}^{\dagger}c_{L-1}c_{L-1}^{\dagger}c_{0}$ whose contributions are identical to the zero-flux partition function's and, nontrivial ones so long as $X$ translates \emph{at least} one fermion winding along the ring to acquire a phase $e^{\pm\text{i}\phi}$.
Note that the most significant nontrivial operator string is to translate one fermion once around the ring as taking the order of $\delta^{L}$.
If $N=2n+1, n\in\mathbbm{Z}_{+}$, the number of fermions been crossed is even, where fermion sign does not arise there.
Therefore, this kind of term takes the form as $(+e^{\text{i}\phi}+e^{-\text{i}\phi})\delta^{L}=+2\delta^{L}\cos\phi$, which maximizes at $\phi=0$.
The higher order nontrivial terms looking like $+2\delta^{2L}\cos(2\phi), +2\delta^{3L}\cos(3\phi), \cdots$ maximize with the same $\phi$.
If $N=2n, n\in\mathbbm{Z}_{+}$, thus the number of fermions been crossed is odd.
Therefore there will be an extra fermion sign arising and this kind of terms take the form as $(-e^{\text{i}\phi}-e^{-\text{i}\phi})\delta^{L}=-2\delta^{L}\cos\phi$, which maximizes at $\phi=\pi$.
The higher order nontrivial contributing terms like $+2\delta^{2L}\cos(2\phi), -2\delta^{3L}\cos(3\phi), \cdots$ which maximize with the same $\phi$.
Once the partition function is maximized, the corresponding free energy is minimized.
\end{proof}

\section{Half-filled spin-$1/2$ free fermions on a non-bipartite odd-numbered ring}
\label{app:half_filled_odd_ring}
Firstly, let us suppose $L = 4n+1, n\in\mathbbm{Z}_{+}$.
$N_{\uparrow}=2n+1, N_{\downarrow}=2n$.
There is a $\phi>0$ threaded through the ring and the corresponding momentum shift is $\phi/L$.
The ground state energy of the fermions around this local minimum can be expressed as
\begin{equation}
    \begin{aligned}
        E_{0\uparrow}&=-2t\cos\left(\frac{\phi}{L}\right)-2t\sum_{l=1}^{n}\cos\left(\frac{2l\pi-\phi}{L}\right)-2t\sum_{l=1}^{n}\cos\left(\frac{2l\pi+\phi}{L}\right), \\
        E_{0\downarrow}&=-2t\cos\left(\frac{\phi}{L}\right)-2t\sum_{l=1}^{n}\cos\left(\frac{2l\pi-\phi}{L}\right)-2t\sum_{l=1}^{n-1}\cos\left(\frac{2l\pi+\phi}{L}\right).
    \end{aligned}
    \label{eq:}
\end{equation}
$E_{0}=E_{0\uparrow}+E_{0\downarrow}$.
$\partial E_{0}/\partial\phi=0$ gives
\begin{equation}
2\sin\left(\frac{\phi}{L}\right)-2\sum_{l=1}^{n}\sin\left(\frac{2l\pi-\phi}{L}\right)+2\sum_{l=1}^{n}\sin\left(\frac{2l\pi+\phi}{L}\right)-\sin\left(\frac{2n\pi+\phi}{L}\right)=0,
\end{equation}
which is accidentally fulled with $\phi=\pi/2$ no matter what $L$ is.
Note that we have the equation
\begin{equation}
    2\sin\left(\frac{\pi}{8n+2}\right)+4\sin\left(\frac{\pi}{8n+2}\right)\sum_{l=1}^{n}\cos\left(\frac{2l\pi}{4n+1}\right)=1,
\end{equation}
which always holds.
Secondly, let us suppose $L = 4n+3, n\in\mathbbm{Z}_{+}$.
$N_{\uparrow}=2n+2, N_{\downarrow}=2n+1$.
Their ground state energies are given by
\begin{equation}
   \begin{aligned}
       E_{0\uparrow}&=-2t\cos\left(\frac{\phi}{L}\right)-2t\sum_{l=1}^{n+1}\cos\left(\frac{2l\pi-\phi}{L}\right)-2t\sum_{l=1}^{n}\cos\left(\frac{2l\pi+\phi}{L}\right), \\
       E_{0\downarrow}&=-2t\cos\left(\frac{\phi}{L}\right)-2t\sum_{l=1}^{n}\cos\left(\frac{2l\pi-\phi}{L}\right)-2t\sum_{l=1}^{n}\cos\left(\frac{2l\pi+\phi}{L}\right).
   \end{aligned}
   \label{eq:}
\end{equation} 
$E_{0}=E_{0\uparrow}+E_{0\downarrow}$.
$\partial E_{0}/\partial\phi=0$ gives
\begin{equation}
    2\sin\left(\frac{\phi}{L}\right)-2\sum_{l=1}^{n}\sin\left(\frac{2l\pi-\phi}{L}\right)+2\sum_{l=1}^{n}\sin\left(\frac{2l\pi+\phi}{L}\right)-\sin\left[\frac{2(n+1)\pi-\phi}{L}\right]=0,
\end{equation}
which is still fulled with $\phi=\pi/2$ no matter what $L$ is.
Note that the equation
\begin{equation}
2\sin\left(\frac{\pi}{8n+6}\right)+4\sin\left(\frac{\pi}{8n+6}\right)\sum_{l=1}^{n}\cos\left(\frac{2l\pi}{4n+3}\right)=1
\end{equation}
always holds.
$\phi=-\pi/2$ goes a similar procedure.
Thus in a word, as long as $L>1$ and $L$ is odd, the optimal fluxes for the ground state of half-filled free spin-$1/2$ fermions are $\pm\pi/2$, which are \emph{independent of $L$ finite or not}.
In Ref.~\cite{lieb1993} we find a similar result while we use different methods.

We can also compute the partition function in the momentum space,
\begin{equation}
\begin{aligned}
Z
&=\text{tr}\left(e^{-\beta H}\right)=\sum_{\alpha}\langle\alpha|e^{2\beta{t}\sum_{\sigma}\sum_{k}\cos{k}c_{k\sigma}^{\dagger}c_{k\sigma}}|\alpha\rangle \\
&=\sum_{\gamma, \eta}\langle\gamma|_{\uparrow}e^{2\beta t\sum_{k}\cos{k}c_{k\uparrow}^{\dagger}c_{k\uparrow}}|\gamma\rangle_{\uparrow}\cdot\langle\eta|_{\downarrow}e^{2\beta t\sum_{k}\cos{k}c_{k\downarrow}^{\dagger}c_{k\downarrow}}|\eta\rangle_{\downarrow} \\
&=\left(\sum_{\gamma}\langle\gamma|e^{2\beta t\sum_{k}\cos{k}c_{k}^{\dagger}c_{k}}|\gamma\rangle\right)\cdot\left(\sum_{\eta}\langle\eta|e^{2\beta t\sum_{k^{\prime}}\cos{k^{\prime}}c_{k^{\prime}}^{\dagger}c_{k^{\prime}}}|\eta\rangle\right),
\end{aligned}
\end{equation}
where we have written a basis as $|\alpha\rangle=|\gamma\rangle_{\uparrow}\otimes|\eta\rangle_{\downarrow}$, of which the Hilbert space dimension is $D=C_{L}^{N_{\uparrow}}\cdot C_{L}^{N_{\downarrow}}$.

\section{Other fillings on the triangular lattice}
\begin{figure}[!ht]
    \centering
    \includegraphics[width=0.85\textwidth]{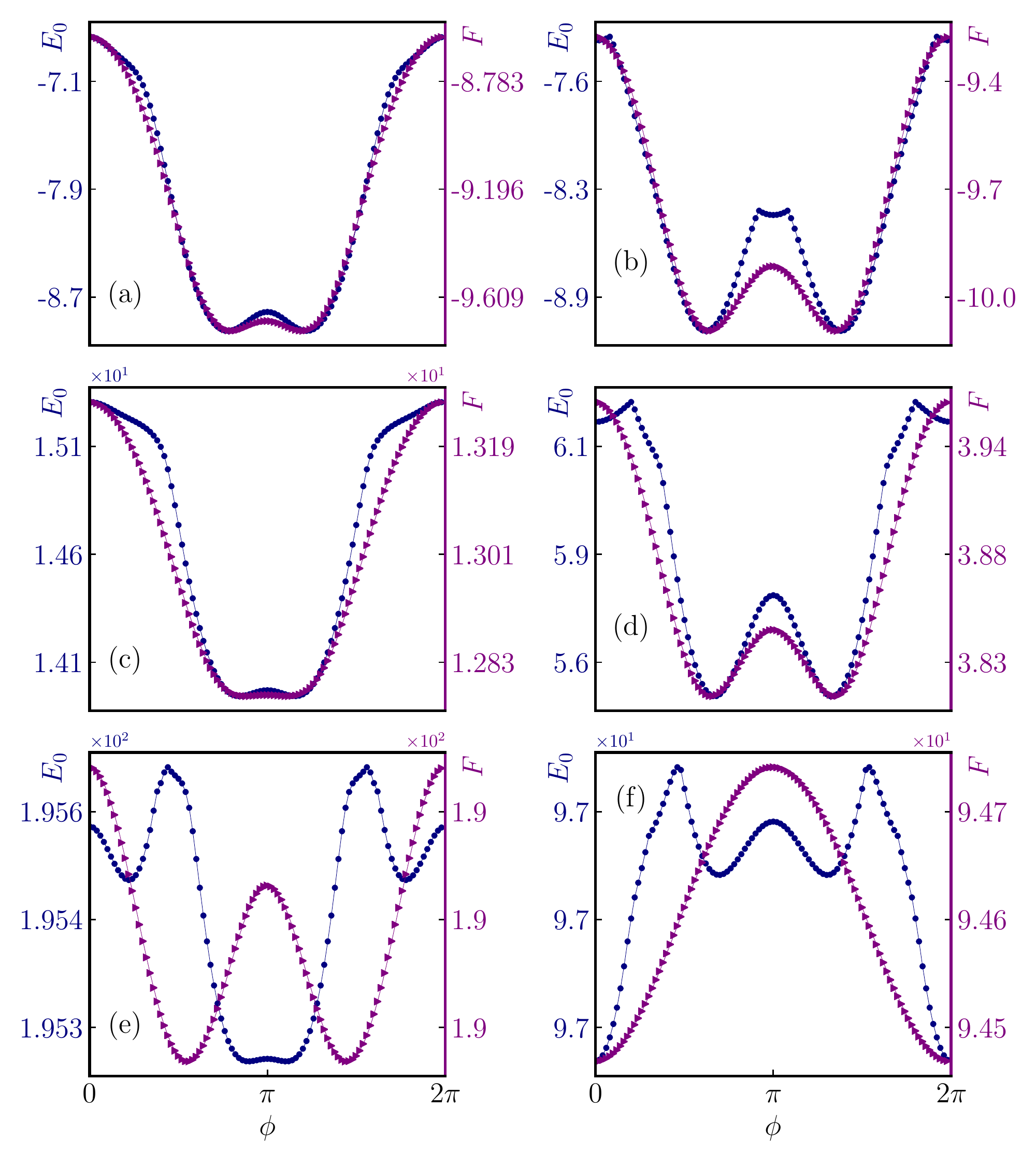}
    \caption{Ground state energy $E_{0}$ and the finite temperature free energy $F$ of the Hubbard model on a triangular lattice consisting of $6$ sites. Vertically, (a, b), (c, d) and (e, f) denote $U/t=0.0, 10.0, 100.0$, respectively. Horizontally, (a, c, e) denote half-filling $N_{\uparrow}=N_{\downarrow}=4$. (b, d, f) denote filling $N_{\uparrow}=4, N_{\downarrow}=3$. Free energy is computed at $\beta=1.0$.}
    \label{fig:flux_triangle_2}
\end{figure}
Here we compute more cases with different fillings of the Hubbard model on the triangular lattice as shown in Figure~\ref{fig:flux_triangle_2}.

\section{Review of the reflection positivity on a bipartite lattice}
\label{app:reveew_RP}
Following the Ref.~\cite{PhysRevLett.73.2158}, on a bipartite graph $\Lambda$, the kinetic energy can be defined as
\begin{equation}
    K=-\sum_{ij, \sigma}t_{ij}c_{i\sigma}^{\dagger}c_{j\sigma}.
    \label{eq:}
\end{equation}
The hopping amplitude satisfies $t_{ij}=t_{ji}^{*}$ thus the hopping matrix is Hermitian $T=T^{\dagger}$.
And the Hubbard term
\begin{equation}
    W=\sum_{j}U_{j}\left( n_{j\uparrow}- \frac{1}{2} \right)\left( n_{j\downarrow}- \frac{1}{2} \right).
    \label{eq:}
\end{equation}
The Hamiltonian is $H=K+W$. The Hamiltonian can be written as $H=H_{L}+H_{R}+H_{I}$. $H_{I}=t_{lr}c_{l}^{\dagger}c_{r}+t_{rl}c_{r}^{\dagger}c_{l}$. We are at liberty to choose $t_{lr}=t_{rl}$. Particle-hole transformation is defined as $\tau c_{i\sigma} \tau^{-1}=c_{i\sigma}^{\dagger}$ and we can find that $\tau( t_{ij}c_{i}^{\dagger}c_{j})\tau^{-1} = t_{ij}c_{i}c_{j}^{\dagger}=-t_{ji}^{*}c_{j}^{\dagger}c_{i}$, which implies
\begin{equation}
    \tau K(T) \tau^{-1}=K(-T^{*}).
    \label{eq:}
\end{equation}
Consider the so called \emph{operator reflection} $\Theta$ combined by three transformations:
\begin{enumerate}
    \item Geometric reflection $\mathcal{R}$.
    \item Particle-hole transformation $\tau$.
    \item Complex conjugation $\mathcal{C}$, which only operates on the complex amplitude.
\end{enumerate}
For instance,
\begin{equation}
    \Theta(t_{ij}c_{i}^{\dagger}c_{j})=\mathcal{C}\left[\tau(t_{i{'}j{'}}c_{i{'}}^{\dagger}c_{j{'}})\tau^{-1} \right]=\mathcal{C}\left(-t_{j{'}i{'}}^{*}c_{j{'}}^{\dagger}c_{i{'}}\right)=-t_{j{'}i{'}}c_{j{'}}^{\dagger}c_{i{'}}. 
    \label{eq:}
\end{equation}

\emph{Reflection positivity}.
Using Trotter expansion we have
\begin{equation}
    Z=\text{tr}\left( e^{-\beta H} \right)=\lim_{M\rightarrow\infty}\text{tr}\left(V^{M}\right)=\lim_{M\rightarrow\infty}\text{tr}\left[\left( V_{I}V_{L}V_{R} \right)^{M}\right],
    \label{eq:}
\end{equation}
where $V_{I}=1-\beta H_{I}/M, V_{L}=\exp\left( -\beta H_{L}/M \right), V_{R}=\exp\left( -\beta H_{R}/L \right)$. Note that $[V_{L}, V_{R}]=0$. Expanding $V^{M}=\sum_{\alpha}X^{\alpha}$, each term has the form $X=a_{0}V_{L}V_{R}a_{1}V_{L}V_{R}\dots a_{M-1}V_{L}V_{R}$. $a_{i}$ can be one of the three items $\mathbbm{1}, c_{l}^{\dagger}c_{r}, -c_{l}{c}_{r}^{\dagger}$. Our strategy is to move all the left operators to the left \emph{without changing the order of the left operators among themselves}. On of the major difficulities here is that $c_{l}^{\#}$ operators have to move through $c_{r}^{\#}$s.

Because of particle number conservation ($V_{L}, V_{R}$ already conserve the particle on each side), the number of factor $c_{l}^{\dagger}c_{r}$ must be equal to the number $c_{l}c_{r}^{\dagger}$, otherwise $\text{tr}(X)=0$. In another word, the density matrix can be represented in the particle-hole symmetric reduced sub-Hilbert space. Denote the number of pairs $c_{l}^{\dagger}c_{r}, c_{l}c_{r}^{\dagger}$ in the sequence $X$ as $N$. The first $c_{l}^{\#}$ moves through zero $c_{r}^{\#}$. The second moves through one. Thus the total number of induced fermion sign is $0+1+2\cdots+(2N)=2N^{2}$, which cancels the fermion sign.

X can be rewritten as $X=X_{L}\otimes X_{R}$. Then $\text{tr}(X)=\text{tr}(X_{L})\cdot\text{tr}(X_{R})$. And $\text{tr}(X_{L})^{*}=\text{tr}\left[ \Theta(X_{L}) \right]$ since particle-hole transformation will not change the Hamiltonian.
Thus we have $|\text{tr}(X_{L})|^{2}=\text{tr}(X_{L})\cdot\text{tr}(X_{L})^{*}=\text{tr}(X_{L})\cdot\text{tr}\left[ \Theta(X_{L}) \right]=\text{tr}[X_{L}\otimes \Theta(X_{L})]$. In the end,
\begin{equation}
    \begin{aligned}
    \lvert\text{tr}\left( V^{M} \right)\rvert^{2}
    &=\left\lvert\sum_{\alpha}\text{tr}(X^{\alpha})\right\rvert^{2}=\left\lvert\sum_{\alpha}\text{tr}(X_{L}^{\alpha})\cdot\text{tr}(X_{R}^{\alpha})\right\rvert^{2} \\
    &\leqslant\sum_{\alpha}\left\lvert\text{tr}(X_{L}^{\alpha})\right\rvert^{2}\sum_{\beta}\left\lvert\text{tr}(X_{R}^{\beta})\right\rvert^{2} \\
    &=\sum_{\alpha}\text{tr}\left[X_{L}^{\alpha}\otimes\Theta(X_{L}^{\alpha})\right]\sum_{\beta}\text{tr}\left[X_{R}^{\beta}\otimes\Theta(X_{R}^{\beta})\right].
    \end{aligned}
    \label{eq:}
\end{equation}
Then we have
\begin{lemma}
    For each $\beta\geqslant 0$ with fixed $K_{I}$,
   \begin{equation}
       Z(H_{L}, H_{R})^{2}\leqslant Z[H_{L}, \Theta(H_{L})]\cdot Z[H_{R}, \Theta(H_{R})].
       \label{eq:}
   \end{equation}
\end{lemma}
\begin{theorem}
    Assume $|t_{ij}|$ are $\Theta$ reflection invariant. $Z$ is maximized by putting flux $\pi$ in each square face of $\Lambda$.
\end{theorem}

\bibliographystyle{plain}
\bibliography{ref}
\end{document}